\documentclass[a4paper,UKenglish,cleveref, autoref, thm-restate, numberwithinsect]{lipics-v2021}

\usepackage{pgfplots}

\hideLIPIcs

\title{Regular Separability in B\"{u}chi VASS}

\author{Pascal Baumann}{Max Planck Institute for Software Systems (MPI-SWS), Germany}{pbaumann@mpi-sws.org}{https://orcid.org/0000-0002-9371-0807}{}

\author{Roland Meyer}{TU Braunschweig, Germany}{roland.meyer@tu-bs.de}{https://orcid.org/0000-0001-8495-671X}{}

\author{Georg Zetzsche}{Max Planck Institute for Software Systems (MPI-SWS), Germany}{georg@mpi-sws.org}{https://orcid.org/0000-0002-6421-4388}{}

\authorrunning{P. Baumann, R. Meyer, and G. Zetzsche} %

\Copyright{Pascal Baumann, Roland Meyer, and Georg Zetzsche} %



\nolinenumbers %

\EventEditors{Petra Berenbrink, Mamadou Moustapha Kant\'{e}, Patricia Bouyer, and Anuj Dawar}
\EventNoEds{4}
\EventLongTitle{40th International Symposium on Theoretical Aspects of Computer Science (STACS 2023)}
\EventShortTitle{STACS 2023}
\EventAcronym{STACS}
\EventYear{2023}
\EventDate{March 7--9, 2023}
\EventLocation{Hamburg, Germany}
\EventLogo{}
\SeriesVolume{254}
\ArticleNo{45}

\usepackage{newfile}
\usepackage{bm}
\usepackage{amsmath}
\usepackage{amssymb}
\usepackage{tikz}
\usetikzlibrary{arrows,automata,chains,matrix,positioning,scopes,calc,decorations.pathmorphing,shapes.misc,shapes,fit}
\tikzset{mode/.style={font=\scriptsize}}
\tikzset{gadget/.style={->,>=stealth,initial text=,minimum size=7pt,auto,on grid,scale=1,inner sep=1pt,node distance=1cm}}
\tikzset{every state/.style={minimum size=15pt,inner sep=1pt,fill=black!10,draw=black!70,thick}}

\newcommand{\cV}{\mathcal{V}}
\newcommand{\cT}{\mathcal{T}}
\newcommand{\cA}{\mathcal{A}}
\newcommand{\Z}{\mathbb{Z}}
\newcommand{\N}{\mathbb{N}}
\newcommand{\Q}{\mathbb{Q}} %
\newcommand{\ourbold}[1]{\bm{#1}}
\newcommand{\br}{\ourbold{r}}
\newcommand{\bmm}{\ourbold{m}}
\newcommand{\bs}{\ourbold{s}}
\newcommand{\bt}{\ourbold{t}}
\newcommand{\be}{\ourbold{e}}
\newcommand{\bx}{\ourbold{x}}
\newcommand{\by}{\ourbold{y}}
\newcommand{\bz}{\ourbold{z}}
\newcommand{\bu}{\ourbold{u}}

\newcommand{\bb}{\ourbold{b}}
\newcommand{\bd}{\ourbold{d}}
\newcommand{\bA}{\ourbold{A}}
\newcommand{\bC}{\ourbold{C}}
\newcommand{\bzero}{\ourbold{0}}

\newcommand{\prof}[1]{\Pi(#1)}

\newcommand{\limsep}[2]{#1\mathrel{|_{\mathsf{lim}}#2}}
\newcommand{\regsep}[2]{#1\mathrel{|}#2}
\newcommand{\notregsep}[2]{#1\mathrel{\not|}#2}
\newcommand{\reset}{\mathsf{r}}

\newcommand{\KM}[1]{\mathsf{KM}(#1)}

\DeclareMathOperator\prefix{prefix} %
\DeclareMathOperator\infix{infix}

\newcommand{\eff}{\Delta}
\DeclareDocumentCommand{\inteff}{o}{%
        \ensuremath{%
                \IfNoValueTF{#1}{%
			\delta%
                }{
			\delta_{#1}%
                }
        }
}

\DeclareDocumentCommand{\exteff}{o}{%
        \ensuremath{%
                \IfNoValueTF{#1}{%
			\varphi%
                }{
			\varphi%
                }
        }
}

\newcommand{\autsteps}{\xrightarrow{*}}

\newcommand{\impl}[2]{``\labelcref{#1}$\Rightarrow$\labelcref{#2}''}

\newcommand{\colorA}{red!60}
\newcommand{\colorB}{blue!60}

\newcommand{\PSPACE}{\mathsf{PSPACE}}
\newcommand{\EXPSPACE}{\mathsf{EXPSPACE}}

\newcounter{inlineenum}
\newcommand{\myitem}{\refstepcounter{inlineenum}(\roman{inlineenum})~}
\newenvironment{myenum}{\setcounter{inlineenum}{0}}{}

\usepackage{todonotes}
\newcommand{\gz}[1]{}
\newcommand{\pb}[1]{}
\newcommand{\rfm}[1]{}

\newcommand{\pump}[1]{#1_{\mathsf{pump}}}

\funding{\flag[3cm]{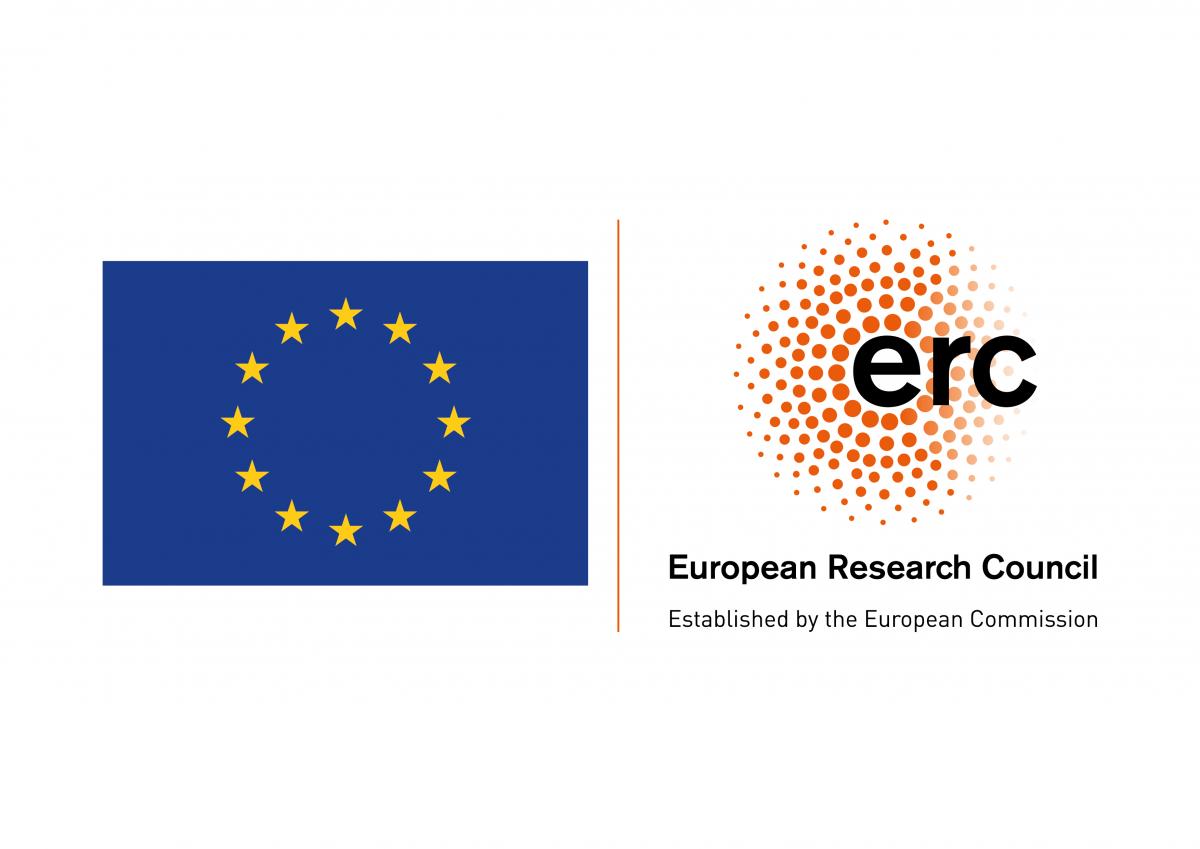}Funded by the European Union (ERC, FINABIS, 101077902). Views and opinions expressed are however those of the author(s) only and do not necessarily reflect those of the European Union or the European Research Council Executive Agency. Neither the European Union nor the granting authority can be held responsible for them.
The second author was supported by the DFG project EDS@SYN: Effective Denotational Semantics for Synthesis.}

\begin{document}

\maketitle
\begin{abstract}
We study the ($\omega$-)regular separability problem for B\"uchi VASS languages: Given two B\"uchi VASS with languages $L_1$ and $L_2$, check whether there is a regular language that fully contains $L_1$ while remaining disjoint from $L_2$.
We show that the problem is decidable in general and \textsf{PSPACE}-complete in the 1-dimensional case, assuming succinct counter updates.
The results rely on several arguments.
We characterize the set of all regular languages disjoint from $L_2$.
Based on this, we derive a (sound and complete) notion of inseparability witnesses, non-regular subsets of $L_1$. 
Finally, we show how to symbolically represent inseparability witnesses and how to check their existence.

\end{abstract}

\newpage %

\section{Introduction}
The separability problem asks, given languages $L_1$ and $L_2$, whether there
exists a language $R$ that \emph{separates} $L_1$ and $L_2$, meaning
$L_1\subseteq R$ and $R\cap L_2=\emptyset$. Here, $R$ is constrained to
be from a particular class $\mathcal{S}$ of admitted separators. Since safety verification
of systems with concurrent components is usually phrased as an intersection problem for finite-word languages, and separators
certify disjointness, deciding separability can be viewed as 
synthesizing safety certificates.
Analogously, deciding separability for infinite-word languages is a way of certifying
liveness.
If $\mathcal{S}$ is the class of ($\omega$-)regular languages, we speak of \emph{regular separability}.

Separability problems have been studied intensively over the last few years.
If the input languages are themselves regular and
$\mathcal{S}$ is a
subclass~\cite{DBLP:journals/corr/PlaceZ14,DBLP:conf/stacs/PlaceZ15,DBLP:conf/csl/PlaceZ14,DBLP:conf/fsttcs/PlaceRZ13,DBLP:conf/icalp/PlaceZ18,DBLP:conf/lics/PlaceZ19,DBLP:journals/tcs/Masopust18,DBLP:conf/icalp/CzerwinskiMM13},
then separability generalizes the classical subclass membership problem. Moreover, separability for languages of infinite-state
systems has received a significant amount of attention~\cite{CzerwinskiZetzsche2020a,CzerwinskiMartensRooijenZeitounZetzsche2017a,WSTSRegSep2018,CzerwinskiL17,ClementeCLP17a,ClementeCLP17b,CzerwinskiHofmanZetzsche2018a,DBLP:conf/concur/AbdullaADK20,Zetzsche2018a,ThinniyamZetzsche2019a,DBLP:conf/icalp/Clemente0P20,DBLP:journals/tcs/ChoffrutDV07}.
Let us point out two prominent cases.

First, one of the main open problems in this line of research is whether
regular separability is decidable for (reachability) languages of \emph{vector addition
systems with states} (VASS): A VASS consist of finitely many control states and
a set of counters that can be \emph{incremented} and \emph{decremented}, but
not tested for zero. Moreover, each transition is labeled by a word over the
input alphabet. Here, a run is accepting if it reaches a final state with all
counters being zero. While there have been several decidability results for
subclasses of the VASS
languages~\cite{CzerwinskiZetzsche2020a,WSTSRegSep2018,CzerwinskiL17,ClementeCLP17a,ClementeCLP17b},
the general case remains open. Second, a surprising result is that if $K$ and
$L$ are coverability languages of \emph{well-structured transition systems} (WSTS), then $K$
and $L$ are separable by a regular language if and only if they are
disjoint~\cite{WSTSRegSep2018}.
As VASS are one example of WSTS, this result also applies to their coverability languages.

\subparagraph{Regular separability in B\"{u}chi VASS} In this paper, we study
the regular separability problem for B\"{u}chi VASS. These are VASS that
accept languages of \emph{infinite words}. A run is accepting if it
visits some final state infinitely often. 
Since no condition is placed on the counter values, Büchi VASS languages are an infinite-word analogue of 
finite-word \emph{coverability languages}, where acceptance is defined by the reached state (not the counters).
The \emph{regular separability problem} is to decide, given B\"{u}chi VASS $\cV_1$ and $\cV_2$, whether there exists an $\omega$-regular language $R$ such that
$L(\cV_1)\subseteq R$ and $L(\cV_2)\cap R=\emptyset$.

Our main results are that (i) regular separability for B\"{u}chi VASS is decidable,
and that (ii) for one-dimensional B\"{u}chi VASS (i.e.\ those with a single counter)
the problem is $\PSPACE$-complete. 
Here, we assume that the counter updates are encoded in binary.

Given that B\"{u}chi VASS accept using final states and their transition systems are WSTS, one may suspect that there is an analogue of the aforementioned result for WSTS: Namely, that two languages of B\"{u}chi VASS are separable by an $\omega$-regular language if and only if they are disjoint. 
We show that this is not the case: There are B\"{u}chi VASS $\cV_1$ and $\cV_2$ such that $L(\cV_1)$ and $L(\cV_2)$ are disjoint, but not separable by an $\omega$-regular language.
In fact, we show an even larger disparity
between these two problems for WSTS in the infinite-word case: 
	We exhibit a natural class of WSTS for which intersection
is decidable but regular separability is not.
Thus, regular separability for B\"{u}chi VASS requires significantly new ideas and involves several phenomena that do not occur for finite-word languages of VASS.

\subparagraph{New phenomena and key ingredients} 
We first observe that we can assume one input language to be fixed, namely an infinite-word version $D_n$ of the Dyck language.
Then, following the \emph{basic separator} approach from~\cite{CzerwinskiZetzsche2020a}, we identify a small class $\mathcal{B}$ of $\omega$-regular languages such that $L$ is separable from $D_n$ if and only if $L$ is included in a finite union of sets from
$\mathcal{B}$. 
Here, a crucial insight is that a B\"uchi automaton can guarantee disjointness from $D_n$ without knowing exactly when the letter balance crosses zero.
Note that a negative letter balance is the exact condition for non-membership in $D_n$.
In contrast, in the finite word case, there are always separating automata that can tell when zero is crossed~\cite{CzerwinskiZetzsche2020a}.
This insight is also key to the example differentiating disjointness and separability in Büchi VASS, and to the undecidability proof for 
certain WSTS despite decidable disjointness.

We then develop a decomposition of B\"{u}chi VASS languages into \emph{finitely many} pieces, which are induced by what we call \emph{profiles}. 
Inspired by B\"uchi automata, the idea of a profile is to fix the set of transitions that can and have to be taken infinitely often in a run. 
Finding the right generalization to Büchi VASS, however, turned out to be non-trivial. 
Our formulation refers to edges in the Karp-Miller graph, augmented by constraints that guarantee the existence of an accepting run. 
The resulting decomposition has properties similar to the decomposition of VASS languages into run ideals~\cite{DBLP:conf/lics/LerouxS15}, which has been useful for previous separability
procedures~\cite{CzerwinskiZetzsche2020a,CzerwinskiHofmanZetzsche2018a}. 

We associate to each profile a system of linear inequalities and show that separability holds if and only if each of these systems is feasible. 
While this yields decidability, checking feasibility is not sufficient to obtain a $\PSPACE$-upper bound in the one-dimensional case. 
Instead, we use Farkas' Lemma to obtain a dual system of inequalities so that separability fails if and only if one dual system is satisfiable. 
A solution to a dual system yields a pattern in the Karp-Miller graph, called \emph{inseparability flower}, which witnesses inseparability.
Compared to prior witnesses for deciding properties of VASS languages (e.g.
regularity~\cite{DBLP:journals/jcss/Demri13}, language
boundedness~\cite{DBLP:journals/tcs/ChambartFS16}, and other
properties~\cite{DBLP:journals/ijfcs/AtigH11}), inseparability flowers are
quite unusual: they contain a non-linear condition, requiring one vector
to be a scalar multiple of another.

For one-dimensional B\"{u}chi VASS, the condition degenerates into a linear one.
This allows us to translate 
inseparability flowers into particular runs in a two-dimensional VASS subject to additional linear constraints.
Using methods from~\cite{DBLP:journals/jacm/BlondinEFGHLMT21}, this yields a
$\PSPACE$ procedure.

\subparagraph{Related work} It was already shown in 1976 that regular
separability is undecidable for context-free
languages~\cite{SzymanskiW-sicomp76,DBLP:journals/jacm/Hunt82a}.  Over the last decade, there has been
intense interest in deciding regular separability for subclasses of finite-word
VASS reachability languages: The problem is decidable for (i)~reachability languages of
one-dimensional VASS~\cite{CzerwinskiL17}, (ii)~coverability languages of
VASS~\cite{WSTSRegSep2018}, (iii)~reachability languages of Parikh automata~\cite{ClementeCLP17b}, and
(iv)~commutative reachability languages of VASS~\cite{ClementeCLP17a}. Moreover, decidability
still holds if one input language is an arbitrary VASS language and the other
is as in (i)-(iii)~\cite{CzerwinskiZetzsche2020a}. As discussed above, for
finite-word coverability languages of WSTS, regular separability is equivalent to
disjointness~\cite{WSTSRegSep2018}.  
Moreover, the aforementioned undecidability for context-free languages has been strengthened to visibly
pushdown languages~\cite{Kopczynski16}. 
To our knowledge, for
languages of infinite words, separability has only been studied for 
regular input
languages~\cite{DBLP:conf/fossacs/PierronPZ16,DBLP:conf/fsttcs/Hugenroth21}.

Our result makes use of Farkas' Lemma to demonstrate the absence of what can be understood as a linear ranking function (on letter balances). 
There are precursors to this. 
In liveness verification~\cite{PR04}, Farkas' Lemma has been used to synthesize, in a complete way, linear ranking functions proving the termination of while programs over integer variables.
In the context of separability for finite words, Farkas' Lemma was used to distinguish separable from non-separable instances~\cite{CzerwinskiZetzsche2020a}, similar to our approach.
The novelty here is the combination of Farkas' Lemma with the new notion of profiles needed to deal with infinite runs.

The languages of B\"{u}chi VASS have first been studied by
Valk~\cite{valk1983infinite} and (in the deterministic case)
Carstensen~\cite{DBLP:conf/mfcs/Carstensen88}.  Some complexity results (such
as $\EXPSPACE$-complexity of the emptiness problem) were shown by
Habermehl~\cite{Habermehl97}. More recently, there have been several papers on
the topological complexity of B\"{u}chi VASS languages (and
restrictions)~\cite{DBLP:journals/mscs/Finkel06,DBLP:books/daglib/p/DuparcFR14,DBLP:journals/fuin/FinkelS21}.
See the recent article by Finkel and
Skrzypczak~\cite{DBLP:journals/fuin/FinkelS21} for an overview.

\section{Preliminaries}
\subparagraph{Dyck Language}
We use an infinite-word version of the Dyck language over $n$ pairs of matching letters $a_i, \bar a_i$. 
We denote the underlying alphabet by $\Sigma_n := \bigcup_{i=1}^n\{a_i,\bar{a}_i\}$. 
The \emph{Dyck language} contains those infinite words where every occurrence of $\bar{a}_i$ has a matching
occurrence of $a_i$ to its left: 
$D_n := \{w \in \Sigma_n^\omega ~|~ \forall v \in \prefix(w)\colon \forall i \in [1, n]\colon \varphi_i(v) \geq 0\}$. 
Here,  $\varphi_i: \Sigma_n^* \rightarrow \Z$ is the $i$th (\emph{letter}) \emph{balance} function that computes for a given word $w$ the difference $|w|_{a_i} - |w|_{\bar{a}_i}$.
We also use $\varphi(w)$ for the vector $(\varphi_1(w),\ldots, \varphi_n(w)) \in \Z^n$.
\subparagraph{B\"uchi VASS and Automata}
A \emph{B\"uchi vector addition system with states (B\"uchi VASS)} of dimension~$d\in\mathbb{N}$ over alphabet $\Sigma$ is a tuple $\cV=(Q, q_0, T, F)$ consisting of a finite set of states $Q$, an initial state $q_0\in Q$, a set of final states $F\subseteq Q$, and a finite set of transitions $T \subseteq\ Q\times \Sigma^*\times \mathbb{Z}^d\times Q$. 
The size of the B\"uchi VASS is $|\cV|:=|Q|+1+|F|+\sum_{(q, w,\delta, q')\in T}|w|+\sum_{i=1}^d\max\{\log |\delta(i)|, 1\}$.
If $d=0$, we call $\cV$ a \emph{B\"uchi automaton}.

The semantics of the B\"uchi VASS is defined over \emph{configurations}, which are elements of $Q\times\mathbb{N}^d$. 
We call the second component in a configuration the \emph{counter valuation} and refer to the entry in dimension $i$ as the \emph{value of counter $i$}.
The \emph{initial configuration} is $(q_0, \bzero)$.  
We lift the transitions of the B\"uchi VASS to a relation over configurations
$\rightarrow\ \subseteq\ Q\times\mathbb{N}^d\times\Sigma^*\times Q\times\mathbb{N}^d$ as follows: 
$(q, \bmm)\xrightarrow{w}(q', \bmm')$ if there is $(q, w, \delta, q')\in T$ so that $\bmm'=\bmm+\delta$. 
A \emph{run} of the B\"uchi VASS is an infinite sequence of transitions of the form
$(q_0, \bzero)\xrightarrow{w_1}(q_1, \bmm_1)\xrightarrow{w_2}\cdots$
Thus, the sequence starts in the initial configuration and makes sure the target of one transition is the source of the next.
The run is \emph{accepting} if it visits final states infinitely often, meaning there are infinitely many configurations $(q, \bmm)$ with $q\in F$. 
The run is said to be \emph{labeled} by the word $w=w_0w_1\cdots$ in $\Sigma^{\omega}$. 
The \emph{language} $L(\cV)$ of the B\"uchi VASS consists of all infinite
words that label an accepting run.
Note that we can always ensure that every accepting run has an infinite-word
label, by tracking in the state whether a non-$\varepsilon$-transition
has occurred since the last visit to a final state.
An infinite-word language is ($\omega$-)\emph{regular}, if it is the language of a B\"uchi automaton. 
As we only consider infinite-word languages, we just call them languages. 

\subparagraph{Karp-Miller Graphs}
We work with the Karp-Miller graph $\KM{\cV}$ associated with a B\"uchi VASS $\cV$~\cite{KarpMiller1969}.
Since we are interested in infinite runs, we define the Karp-Miller graph as a B\"uchi automaton.
Its state set is a finite set of \emph{extended configurations}, which are elements of $Q\times (\mathbb{N}\cup\{\omega\})^d$.
The initial state is the initial configuration in the B\"uchi VASS. 
The final states are those extended configurations $(q, \bmm)$ with $q\in F$. 
The transitions are labeled by $T$, so instead of letters they carry full B\"uchi VASS transitions. 
An entry $\omega$ in an extended configuration denotes the fact that a prefix of a run can be repeated to produce arbitrarily high counter values. 
More precisely, the Karp-Miller graph is constructed as follows.
From an extended configuration $(q, \bmm)$ we have a transition labeled by $(q_1, a, \delta, q_2)$, if $q=q_1$ and $\bmm+\delta$ remains non-negative. 
The latter addition is defined componentwise and assumes $\omega+k:=\omega=:k+\omega$ for all $k\in\Z$.
The result of taking the transition is the extended configuration $(q_2, \bmm_2)$, where $\bmm_2$ is constructed from $\bmm+\delta$ as follows.
We raise to $\omega$ all counters $i$ for which there is an earlier configuration $(q_2, \bmm_1)$ with %
$\bmm_1\leq \bmm+\delta$ and $\bmm_1(i)<[\bmm+\delta](i)$, earlier meaning on some path from $(q_0, \bzero)$ to $(q, \bmm)$. 
If this is the case, the path from $(q_2, \bmm_1)$ to $(q_2, \bmm+\delta)$ can be repeated indefinitely to produce arbitrarily high values for counter $i$.
We refer to the repetition of such a path in a run as \emph{pumping}. 

The Karp-Miller graph over-approximates the language of the B\"uchi VASS in the following sense.
Every infinite sequence of transitions that leads to a run of the B\"uchi VASS is the labeling of an infinite run in the Karp-Miller graph.
Moreover, if the run of the B\"uchi VASS is accepting, so is the run in the Karp-Miller graph.
In the other direction, every finite transition sequence in the Karp-Miller graph represents a transition sequence in the B\"uchi VASS. 
The sequence in the B\"uchi VASS, however, may be longer to compensate negative effects on $\omega$-entries by pumping.

\section{Problem, Main Result, and Proof Outline}\label{Section:Outline}
\newcommand{\sep}{\mathop{|}}
A language $R$ is a \emph{regular separator} for a pair of languages $L_1, L_2$, if $R$ is regular, $L_1\subseteq R$, and $R\cap L_2=\emptyset$. 
We write $L_1\sep L_2$ for the fact that a regular separator exists.
We consider here languages of B\"uchi VASS, and formulate the \emph{regular separability problem} as follows. 
Given B\"uchi VASS $\cV_1$, $\cV_2$, check whether $L(\cV_1)\sep L(\cV_2)$ holds.
Our main result is the following.
\begin{theorem}\label{Theorem:MainResult}
The regular separability problem for B\"uchi VASS is decidable.
\end{theorem}
It should be noted that our procedure is non-primitive recursive, as it explicitly constructs the Karp-Miller graph of an input B\"{u}chi VASS, which can be of Ackermannian size~\cite[Theorem~2]{mayr1981complexity}.
In the case of VASS coverability languages (and even for more general WSTS), it is known that regular separability is equivalent to disjointness~\cite{WSTSRegSep2018}. 
Thus, for finite words, separability reduces to the much better understood problem of disjointness. 
For the infinite-word languages considered here, the situation is different. 
\begin{theorem}\label{Theorem:CounterExamples}
There are B\"uchi VASS languages $L_1$, $L_2$ with $L_1\cap L_2= \emptyset$ and $L_1\not\sep L_2$. 
There are classes of WSTS where intersection is decidable but separability is not.
\end{theorem}
For the second statement, we introduce the class of \emph{weak B\"{u}chi reset VASS},
which are VASS with reset instructions, with the additional constraint that each run can only use resets a finite number of times. Details can be found in \cref{appendix-intersection}.

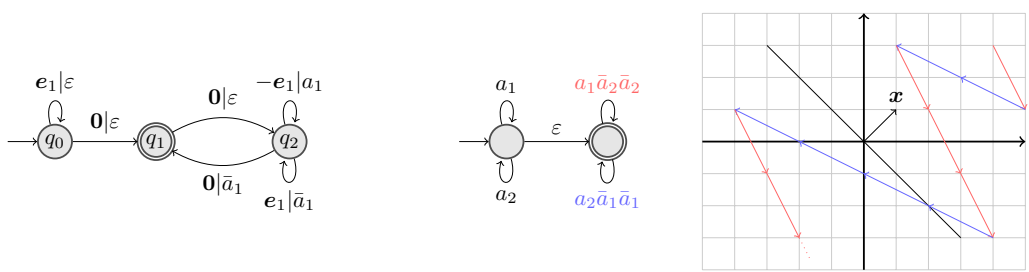
\begin{figure}[t]
\begin{center}
\scalebox{0.85}{
\begin{tikzpicture}[initial text={},baseline]
\node[state,initial] (q0) {$q_0$};
\node[state,right=1cm of q0,accepting] (q1) {$q_1$};
\node[state,right=1.5cm of q1] (q2) {$q_2$};
\path[->] 
(q0) edge [loop above] node {$\be_1|\varepsilon$} (q0)
(q0) edge node[above] {$\bzero|\varepsilon$} (q1)
(q1) edge[bend left] node[above] {$\bzero|\varepsilon$} (q2)
(q2) edge[bend left] node[below] {$\bzero|\bar{a}_1$} (q1)
(q2) edge[loop above] node {$-\be_1|a_1$} (q2)
(q2) edge[loop below] node {$\be_1|\bar{a}_1$} (q2)
;
\end{tikzpicture}
\hspace{1.5cm}
\begin{tikzpicture}[initial text={},baseline]
\node[state,initial] (q0) {};
\node[state,accepting,right=1cm of q0] (q1) {};
\path[->] 
(q0) edge [loop above] node {$a_1$} (q0)
(q0) edge [loop below] node {$a_2$} (q0)
(q0) edge [above] node {$\varepsilon$} (q1)
(q1) edge [loop above] node {$\textcolor{\colorA}{a_1\bar{a}_2\bar{a}_2}$} (q1)
(q1) edge [loop below] node {$\textcolor{\colorB}{a_2\bar{a}_1\bar{a}_1}$} (q1);

\end{tikzpicture}
\pgfplotsset{compat=1.11}
\newcommand{\steps}[4]{%
	\foreach \x in {0,...,#3}
		\path[draw,->,#4] ($#1 + \x*#2$) -- ($#1 + #2 + \x*#2$);
}
\newcommand{\xmax}{5}
\newcommand{\xmin}{-5}
\newcommand{\ymax}{4}
\newcommand{\ymin}{-4}
\begin{tikzpicture}[scale=0.5, baseline]
\draw[help lines,gray!40] (\xmin,\ymin)grid(\xmax,\ymax);
\draw[thick,->] (\xmin,0)-- (\xmax,0);
\draw[thick,->] (0,\ymin)-- (0,\ymax);
\draw[-] (-3,3) -- (3,-3);
\steps{(4,3)}{(1,-2)}{0}{\colorA}
\steps{(5,1)}{(-2,1)}{1}{\colorB}
\steps{(1,3)}{(1,-2)}{2}{\colorA}
\steps{(4,-3)}{(-2,1)}{3}{\colorB}
\steps{(-4,1)}{(1,-2)}{1}{\colorA}
\draw[dotted,color=\colorA, shorten >=0.7cm] (-2,-3) -- (-1,-5);
\coordinate (x) at (1,1);
\draw[->] (0,0) -- (x);
\node[above=0cm of x] {$\bx$};

\end{tikzpicture}
}%
\end{center}
\vspace{-0.8cm}
\caption{Left: A B\"{u}chi VASS accepting a language $S$ with $S\cap D_1=\emptyset$ but $\notregsep{S}{D_1}$. Here, $\be_1\in\Z$ is the one-dim.\ vector with entry $1$.
Right: A regular language that is not included in a finite union of languages $P_{i,k}$ and $S_{i,k}$, but that is included in $S_{\bx, k}$ for $\bx=(1, 1)$, $k = 1$. 
The horizontal and vertical dimensions denote the balance for $a_1$ resp. $a_2$. \label{Figure:examples}}
\end{figure}

For the first statement of \cref{Theorem:CounterExamples}, we give an intuition
and refer to \cref{appendix-problem} for details. We choose $L_1=L(\cV)$, where
$\cV$ is the B\"{u}chi VASS in \cref{Figure:examples}(left), and $L_2=D_1$, the
Dyck language.  To show $\notregsep{L(\cV)}{D_1}$, suppose there is a B\"{u}chi
automaton $\cA$ with $n$ states such that $L(\cV)\subseteq L(\cA)$ and
$L(\cA)\cap D_1=\emptyset$.  Then $\cA$ has to accept
$(a_1^n\bar{a}_1^{n+1})^\omega \in L(\cV)$.  However, pumping yields that for
some $m > n$ the word $(a_1^m\bar{a}_1^{n+1})^\omega \in D_1$ also has to be
accepted by $\cA$, contradiction.  Moreover, to show $L(\cV) \cap D_1 =
\emptyset$ we observe that in accepting runs of $\cV$, almost every visit
(meaning: all but finitely many) to the final state drops the letter balance by
$1$.  Therefore on any accepting run this balance eventually becomes negative,
yielding a word outside of $D_1$.

In the remainder of the section, we outline the proof of Theorem~\ref{Theorem:MainResult}.  
Assume we are given $L_1=L(\cV_1)$ and $L_2=L(\cV_2)$ and this is a non-trivial instance of separability, meaning $L_1, L_2$ are not regular and $L_1\cap L_2=\emptyset$.  
For proving separability, we could enumerate regular languages until we find a separator. 
The difficult part is disproving separability. 
Inseparability of $L_1$ and $L_2$ is witnessed by a set of words $W\subseteq L_1$ so that every regular language $R$ containing them already intersects $L_2$, formally: $W\subseteq R$ implies $R\cap L_2\neq\emptyset$. 
Showing the existence of such a set $W$ is difficult for two reasons. 
First, it is unclear which sets of words ensure the universal quantification over all regular languages. 
Second, as we have a non-trivial instance of separability, $W$ (if it exists) will be a non-regular language. 
So it is unclear how to represent it in a finite way and how to check its existence. 

To address the first problem and understand the sets of words that disprove separability, we use diagonalization.
Call an \emph{($L_2$-)separator candidate} a regular language that is disjoint from~$L_2$. 
Let $R_1, R_2, \ldots$ be an enumeration of the separator candidates.
If $L_1$ is not separable from $L_2$, for every $R_i$ there is a word $w_i\in L_1$ with $w_i\notin R_i$. 
We call such a set of words $W=\{w_1, w_2, \ldots\}$ that escapes every separator candidate an \emph{inseparability witness}. 
\begin{observation}
$L_1\not\sep L_2$ if and only if there is an inseparability witness.
\end{observation}

Our decision procedure will check the existence of an inseparability witness. 
We obtain the procedure in four steps: the first is a simplification, the second is devoted to understanding the separator candidates, the third is another simplification, and the last characterizes the inseparability witnesses and checks their existence.
\subparagraph{Step 1: Fixing $\boldsymbol{L_2}$}
We first reduce general regular separability to regular separability from the Dyck language.
The reduction is simple and works just as for finite words~\cite{CzerwinskiZetzsche2020a}.
\begin{restatable}{lemma}{oneLanguageFixed}\label{one-language-fixed}
Given B\"uchi VASS $\cV_1$ and $\cV_2$, we can compute a B\"uchi VASS $\cV$ over $\Sigma_n$ so that $L(\cV_1)\sep L(\cV_2)$ if and only if $L(\cV)\sep D_n$, where $n$ is the dimension of $\cV_2$.
\end{restatable}

\subparagraph{Step 2: Understanding the Separator Candidates}

To understand the regular languages that are disjoint from $D_n$, we will define 
\emph{basic separators}, sets $P_{i, k}$ and $S_{\bx, k}$, on which we elaborate in a moment. 
The following theorem says that finite unions of basic separators are sufficient for regular separability.
This is our first technical result and shown in \cref{sec:basic-separators}. 
\begin{theorem}\label{thm-basic-separators}
	If $R\subseteq\Sigma_n^\omega$ is regular and $R\cap D_n=\emptyset$, then $R$ is included in a finite union of basic separators.
\end{theorem}

For the definition of $P_{i, k}$, we note that the words outside $D_n$ have, for some index $i\in[1, n]$, an earliest moment in time where the balance between $a_i$ and $\bar a_i$ falls below zero.   
To turn this into a regular language, we impose an upper bound $k\in \mathbb{N}$ on the (positive) balance between the letters $a_i$ and $\bar a_i$ that is maintained until the earliest moment is reached.  
This yields the regular language 
\begin{align*}
  P_{i,k}\ :=\  \{w \in \Sigma_n^\omega \mid \exists v \in \prefix(w)\colon \varphi_i(v) < 0 \wedge \forall u \in \prefix(v)\colon \varphi_i(u) \leq k\}.
\end{align*}

The family of languages $P_{i, k}$ already captures the complement of $D_n$.
The problem is that we may need infinitely many such languages to cover the language $R$ of interest. 
For every bound $k$, a regular $R$ with $R\cap D_1=\emptyset$ may contain a word with a higher balance before falling below zero, take for example $R=a_1^*\bar{a}_1^\omega$.
The first insight is that if $R$ can fall below zero from arbitrarily high values, 
then the underlying B\"uchi automaton has to contain loops with a negative balance. 
The $R$ thus contains words $uv$ with an unconstrained prefix and a suffix that decomposes into $v= v_1v_2\cdots$ so that every infix $w=v_\ell$ has a negative balance on letter $a_i$. 
The observation suggests the definition of a language that contains precisely the words $u.v$. 
To make the language regular, we impose a bound $k$ on the positive balance that can be used during the infixes $w$.
Call the resulting language $S_{i, k}$.  
Unfortunately, taking the $P_{i, k}$ and the $S_{i,k}$ as basic separators is still not enough: 
\cref{Figure:examples}(right) exhibits a regular language, disjoint from $D_1$, that is not included in a finite union of $P_{i, k}$ and $S_{i, k}$, because it contains infixes where the balance on each letter exceeds all bounds in each coordinate.
The second insight is that we can catch the remaining words with a version of $S_{i,k}$ that weights coordinates with some $\bx\in\N^n$.
Let us give some intuition on this.
The words from $R$ that we cannot catch with a $P_{i,k}$ must come across, for each $i$ that becomes negative, a loop with positive balance on $i$ (otherwise, the balance on those $i$ would be bounded). 
But then, the only way such words can avoid $D_1$ is by ending up in a strongly connected component where \emph{every} loop (with a final state) makes progress towards crossing $0$, i.e.\ is negative in some coordinate.
One can then conclude that even all $\Q_{\ge 0}$-linear combinations of loops (a convex set) must avoid the positive orthant $\Q_{\ge 0}^n\subset \Q^n$.
By the Hyperplane Separation Theorem (we use it in the form of Farkas' Lemma), this is certified by a hyperplane that separates all loop effects from $\Q^n_{\ge 0}$.
This hyperplane is given by some orthogonal vector $\bx\in\N^n$, meaning that every loop balance must have negative scalar product with $\bx$. Hence, we can catch these words by:
\begin{align*}
  S_{\bx,k}\ :=\ \left\{u.v \in \Sigma_n^\omega ~\middle|~
  \begin{aligned}
    \text{a.)}~\,&\forall f \in \infix(v)\colon \langle \bx,\varphi(f)\rangle \leq k\text{, and}\\
    \text{b.)}~\,&v=v_0.v_1.v_2\cdots \wedge \forall \ell \in \mathbb{N}\colon \langle\bx,\varphi(v_\ell)\rangle < 0
  \end{aligned}
  \right\}.
\end{align*}
Coming back to \cref{Figure:examples}(right), the weight vector $\bx=(1, 1)$ guarantees that the weighted balance decreases indefinitely and also the weighted balances of all infixes stay bounded.  In~\cite{CzerwinskiZetzsche2020a}, a similar argument has been used to show sufficiency of basic separators.
\subparagraph{Step 3: Pumpable Languages}
With the basic separators at hand, the task is to understand the sets of words witnessing inseparability.
While studying this problem, we observed that the argumentation for the $P_{i, k}$ was always similar to the one for the $S_{\bx, k}$.
This led us to the question of whether we can get rid of the $P_{i, k}$ in separators.
The answer is positive, and hinges on a new notion of pumpability for languages over $\Sigma_n$.

Call infinite words $u$ and $v$ \emph{equivalent}, written 
$u\sim v$, if $v$ can be obtained from $u$ by removing and inserting
finitely many letters: There are $u_0,v_0\in\Sigma^*$ and
$w\in\Sigma^\omega$ such that $u=u_0w$ and $v=v_0w$. 
We say that a language $L\subseteq\Sigma_n^\omega$ is \emph{pumpable} if 
 for every $w\in L$ and every $k\in\N$, there exists a
decomposition $w=w_0w_1$ and a word $w_0'\in\Sigma_n^*$ that is a prefix of a word in $D_n$ such that $w_0'.w_1\in L$
and the letter balance satisfies the following: (a)~$\varphi(w_0')\geq \varphi(w_0)$ and
(b)~for the indices $i\in [1, n]$ where $\varphi_i$ becomes negative on some prefix of $w$, we have 
$\varphi_i(w_0')\ge\max\{\varphi_i(w_0), 0\}+k$. 
The consequence of this definition is that a pumpable language leaves every language $P_k:=\bigcup_{i \in [1,n]} P_{i,k}$. 
Indeed, for every word $w\in L$ and every
$k\in\N$, there is a word $w'\in L$ with $w\sim w'$ where the letter balance exceeds $k$ before becoming negative, and thus $w'\notin P_k$.
With the previous characterization of separator candidates, what is left to separate $L$ from $D_n$ are the languages $S_{\bx, k}$.  
\begin{restatable}{lemma}{prefixIndependentLimSup}\label{prefix-independent-limsep}
  If $L\subseteq\Sigma_n^\omega$ is pumpable, then $\regsep{L}{D_n}$ if and only if $\limsep{L}{D_n}$, where $\limsep{L}{D_n}$ means
$L\subseteq \bigcup_{\bx\in X} S_{\bx,k}$ for some finite set $X\subseteq\N^n$
and some $k\in\N$. 
\end{restatable}

In our context, pumpability is interesting because we can turn every B\"uchi VASS language into a pumpable language without affecting separability.

\begin{theorem}\label{make-pumpable}
Let $\cV$ be a $d$-dim.\ B\"{u}chi VASS over $\Sigma_n$.
We can compute a  $d$-dim.\ B\"{u}chi VASS  $\pump{\cV}$ that satisfies the following:
  \begin{enumerate}
    \item $L(\pump{\cV})$ is pumpable,
    \item there is a $k\in\N$ so that $L(\pump{\cV})\subseteq L(\cV)\subseteq L(\pump{\cV})\cup P_{k}$, and
    \item $\regsep{L(\cV)}{D_n}$ if and only if $\regsep{L(\pump{\cV})}{D_n}$.\vspace{0.1cm}
  \end{enumerate}
\end{theorem} 
The construction of $\pump{\cV}$ employs the Karp-Miller graph in an original way, namely to track the unboundedness of letter balances.
Let  $\bar{\cV}$ be the $(d+n)$-dimensional B\"uchi VASS obtained from $\cV$ 
by tracking the effect of the letters from $\Sigma_n$ in  $n$ additional counters.
For $\bar{\cV}$, we construct the Karp-Miller graph. 
The relationship between the languages of $\KM{\bar{\cV}}$ and $\cV$ is as follows.
For all words where every letter balance stays non-negative, their runs in $\cV$ can be mimicked in $\KM{\bar{\cV}}$.
For all other words, where the balance eventually becomes negative, this only holds
if the corresponding counter in $\bar{\cV}$ has been raised to $\omega$ beforehand.
Essentially, the new B\"{u}chi VASS $\pump{\cV}$ restricts $\cV$ to those runs that have counterparts in $\KM{\bar{\cV}}$.
This is achieved with a simple product construction of $\cV$ and $\KM{\bar{\cV}}$.
The thing to note is that every word from $L(\cV)$ that does not make it into $L(\pump{\cV})$ belongs to $P_{k}$, where $k$ is the maximum concrete number in $\KM{\bar{\cV}}$: A run in $\cV$ that cannot be mimicked in $\KM{\bar{\cV}}$ will at some point have a negative letter balance, before reaching $\omega$ in $\KM{\bar{\cV}}$ in that component; thus all counter values had been at most $k$ until that point.

An example on how to construct $\bar{\cV}$ and $\pump{\cV}$ can be found in \cref{Figure:pumpability},
where both were constructed for the B\"{u}chi VASS found in \cref{Figure:examples}(left).

\begin{figure}[t]
  \begin{center}
    \scalebox{0.85}{
      \begin{tikzpicture}[initial text={},baseline]
        \node[state,initial] (q0) {$q_0$};
        \node[state,right=1cm of q0,accepting] (q1) {$q_1$};
        \node[state,right=2.0cm of q1] (q2) {$q_2$};
        \path[->] 
          (q0) edge [loop above] node {$(1,0)|\varepsilon$} (q0)
          (q0) edge node[above] {$\bzero|\varepsilon$} (q1)
          (q1) edge[bend left=15] node[above] {$\bzero|\varepsilon$} (q2)
          (q2) edge[bend left=15] node[below] {$(0,-1)|\varepsilon$} (q1)
          (q2) edge[loop above] node {$(-1,1)|\varepsilon$} (q2)
          (q2) edge[loop below] node[below=.1cm] {$(1,-1)|\varepsilon$} (q2)
        ;
      \end{tikzpicture}
      \hspace{1cm}
      \begin{tikzpicture}[initial text={},baseline]
        \node (q0) {$(q_0,\omega,0)$};
        \node[right=1cm of q0] (q1) {$(q_1,\omega,0)$};
        \node[right=2.0cm of q1] (q2) {$(q_2,\omega,0)$};
        \node[below=1cm of q1] (p1) {$(q_1,0,0)$};
        \node[initial,left=1cm of p1] (p0) {$(q_0,0,0)$};
        \node[right=2.0cm of p1] (p2) {$(q_2,0,0)$};
        \node[above=1cm of q2] (r2) {$(q_2,\omega,\omega)$};
        \node[left=2cm of r2] (r1) {$(q_1,\omega,\omega)$};
        \path[->] 
          (p0) edge node[left] {$\be_1|\varepsilon$} (q0)
          (p0) edge node[below] {$\bzero|\varepsilon$} (p1)
          (p1) edge node[below] {$\bzero|\varepsilon$} (p2)
          (q0) edge [loop above] node {$\be_1|\varepsilon$} (q0)
          (q0) edge node[below] {$\bzero|\varepsilon$} (q1)
          (q1) edge node[below] {$\bzero|\varepsilon$} (q2)
          (q2) edge node[left] {$-\be_1|a_1$} (r2)
          (r1) edge[bend left=10] node[above] {$\bzero|\varepsilon$} (r2)
          (r2) edge[bend left=10] node[below] {$\bzero|\bar{a}_1$} (r1)
          (r2) edge[loop above] node {$\begin{aligned}
              -\be_1|a_1 \\
              \be_1|\bar{a}_1
            \end{aligned}$} (r2)
        ;
      \end{tikzpicture}
    }%
  \end{center}
  \vspace{-0.5cm}
  \caption{Left: The B\"{u}chi VASS $\bar{\cV}$ constructed from the B\"{u}chi VASS $\cV$ found in \cref{Figure:examples}(left).
  Note how the added second counter tracks the letter balance of the now removed transition labels, incrementing on letter $a_1$ and decrementing on letter $\bar{a}_1$.
  Right: The B\"{u}chi VASS $\pump{\cV}$ corresponding to $\cV$ as given by \cref{make-pumpable}.
  Here we did not mark the final states to reduce visual clutter; every state that includes $q_1$ is considered final. Similarly, the two labels above the loop in the top right correspond to two distinct transitions.
  Note that $\pump{\cV}$ essentially looks like $\KM{\bar{\cV}}$, just with different transition labels.
  \label{Figure:pumpability}}
\end{figure}

In the proof of Theorem~\ref{thm-basic-separators}, we make use of Theorem~\ref{make-pumpable} (recall that a regular language is the language of a $0$-dimensional B\"uchi VASS). 
This may look like cyclic reasoning, but it is not:
We will show Theorem~\ref{make-pumpable}(1)+(2) directly, using the arguments above.
With this, we prove Theorem~\ref{thm-basic-separators}, which in turn is used to derive
Lemma~\ref{prefix-independent-limsep} and Theorem~\ref{make-pumpable}(3).

%
%
%
%
\subparagraph{Step 4: Non-Separability Witnesses and Decidability}
Because of pumpability, it remains to decide whether a B\"{u}chi VASS language
$L(\cV)$ is included in a finite union $\bigcup_{\bx\in X}S_{\bx,k}$ for some $k$.  
Part of the difficulty is that we have no bound on the cardinality
of $X$.
To circumvent this, we decompose $L(\cV)$ into a finite union 
$\bigcup_\pi L_\pi(\cV)$, where $\pi$ is a \emph{profile}, meaning a set of edges in $\KM{\cV}$
seen infinitely often during a run of $\cV$. 
We then show that each $L_\pi(\cV)$ is either (i)~included in a single separator
$S_{\bx,k}$ or (ii)~escapes every finite union $\bigcup_{\bx\in X} S_{\bx,k}$.

Here, it is key to show an even stronger fact: In case (i), not only $L_\pi(\cV)$
is included in some $S_{\bx,k}$, but the entire set of runs in $\KM{\cV}$ that eventually remain in $\pi$.
The advantage of strengthening is that finiteness of $\KM{\cV}$ allows us to
express inclusion in $S_{\bx,k}$, for some $k$, as a \emph{finite} system of linear inequalities over $\bx$: 
We say that (1)~the balance of every primitive cycle, weighted by $\bx$, 
is at most zero  and (2)~the balance, weighted by $\bx$, of some cycle containing all edges from $\pi$ is negative.
Here, (1) and (2) correspond to Conditions a.) and b.) of $S_{\bx,k}$. If they are met, then the runs of $\KM{\cV}$ along $\pi$ are included in $S_{\bx,k}$ for some $k$.

We then prove that if the system is not feasible, then $\cV$ has runs that
escape every finite union $\bigcup_{\bx\in X} S_{\bx,k}$. 
To this end, we employ Farkas' Lemma: 
It tells us that if there is no solution, then the dual system has a solution. 
The solution of the dual system can be interpreted as an executable linear combination of primitive
cycles with non-negative balances. 
We show that these cycles can be arranged in a pattern in $\KM{\cV}$ we call \emph{inseparability flower}.
Such an inseparability flower then yields a sequence of runs $\rho_1,\rho_2,\ldots$ in $\KM{\cV}$ such that $\rho_k$ 
escapes $S_{\bx,k}$ for every vector $\bx$.
Finally, pumpability allows us to lift these runs of $\KM{\cV}$ to runs of $\cV$
and thus conclude inseparability.

This equips us with two possible decision procedures: We can either check solvability
of each system of inequalities, or detect inseparability flowers in $\KM{\cV}$.

\section{Basic Separators}\label{sec:basic-separators}
We prove \cref{thm-basic-separators}, that any regular language $R$ over $\Sigma_n$ with $R\cap D_n=\emptyset$ is contained in a finite union of languages $P_{i, k}$ and $S_{\bx, k}$. 
Note that a single value of $k$ is sufficient, since we have $P_{i,k} \subseteq P_{i,k+1}$ and $S_{\bx,k} \subseteq S_{\bx,k+1}$ for each $i,\bx,k$.
The proof decomposes the B\"uchi automaton for $R$ in a way that allows us to forget about connectedness issues and reason over cycles (and their letter balances) using techniques from linear algebra. 
We make use of the following basic fact from linear programming~\cite[Corollary 7.1f]{Schr86}.
\begin{theorem}[Farkas' Lemma (variant), \cite{Schr86}]\label{farkas-lemma}
  Let $\bA \in \Q^{m \times n}$ be a matrix and let $\bb \in \Q^m$ be a vector. Then the system $\bA \bx \leq \bb$ has a solution $\bx \in \Q_{\geq 0}^n$ if and only if $\by^\top \bb \geq 0$ for each vector $\by \in \Q_{\geq 0}^m$ with $\by^\top \bA \geq \bzero$.
\end{theorem}

\subparagraph{Decomposing with profiles}
We decompose $R = L(\cA)$ into a (not necessarily disjoint) union of several languages, each linked to a so-called \emph{profile}. 
We will later see that for pumpable $R$, every such profile language already has to be contained in a single $S_{\bx,k}$.
\begin{definition}
  Let $\cA$ be a B\"{u}chi automaton. A \emph{profile of $\cA$} is a set
  $\pi$ of transitions of $\cA$ for which there exists a cycle $\sigma_\pi$ in
  $\cA$ such that (a)~$\sigma_\pi$ contains exactly the transitions in $\pi$,
  and (b)~$\sigma_\pi$ starts (and ends) in a final state $q_\pi$.
\end{definition}
We denote by $\Pi(\cA)$ the finite set of profiles of $\cA$.
Moreover, we associate to every accepting run $\rho$ of $\cA$ its profile $\Pi(\rho)$, which contains exactly the transitions appearing infinitely often in $\rho$.
This definition is sound, as the infinitely occurring transitions of an accepting run must form a cycle due to repetition, which visits a final state due to acceptance.

Given a profile $\pi$ of $\cA$, we define $L_\pi(\cA) \subseteq L(\cA)$ to be the language of all words that have an accepting run $\rho$ of $\cA$ with $\Pi(\rho) = \pi$.
Note that this language is still regular:
From $\cA$ one can construct a B\"{u}chi automaton that guesses a point after which only transitions from $\pi$ can occur, and once this point is reached it keeps a list of already used transitions from $\pi$ in each state.
Then only once all transitions of $\pi$ have been used the state becomes final and the list is set back to empty.

This now allows us to view $R$ as the union of the languages $L_\pi(\cA)$ with $\pi \in \Pi(\cA)$.
We show that each language $L_\pi(\cA)$ is either contained in $S_{\bx,k}$ for some $\bx,k$, or there is a cycle that, assuming the pumpability from the previous section, makes $L_\pi(\cA)$ intersect $D_n$.
\begin{lemma}\label{lem-basic-post-separators}
  Let $\cA$ be a B\"{u}chi automaton over $\Sigma_n$ and let $\pi$ be one of its profiles. Then one of the following conditions holds:
  \begin{romanenumerate}
    \item\label{post-separator-exists} There is a number $k \in \N$ and a vector $\bx \in \N^n$ such that $L_\pi(\cA) \subseteq S_{\bx,k}$, or
    \item\label{post-separator-nonex-cycle} there is a cycle $\sigma'$ in $\cA$ over $w'$ with $\varphi(w') \geq \bzero$, and $\sigma'$ contains all transitions from $\pi$.\vspace{0.1cm}
  \end{romanenumerate}
\end{lemma}
Assume $L_\pi(\cA)\neq \emptyset$, otherwise Condition~(\labelcref{post-separator-exists}) trivially holds. 
We build a system $\bA_\pi\bx \leq \bb$ of linear inequalities as follows. It contains one inequality $\langle\bx, \varphi(v)\rangle \leq 0$ for each word $v$ read by a primitive cycle of transitions in $\pi$. 
By \emph{primitive cycle} we mean a cycle that does not repeat a state. 
Moreover, the system contains the inequality $\langle\bx, \varphi(v_\pi)\rangle \leq -1$ for the
cycle $\sigma_\pi$ over $v_\pi$ that justifies the profile $\pi$. 
Let us quickly remark that the solution space of the system $\bA_\pi\bx \leq \bb$ is independent of the precise choice of the justifying cycle $\sigma_\pi$:
To see this, we claim that $\bA_\pi\bx\leq\bb$ holds if and only if all primitive cyles in $\pi$ have an $\bx$-weighted balance at most zero, and at least one primitive cycle in $\pi$ has a strictly negative $\bx$-weighted balance.
For the ``if'' direction, note that a sufficiently long repetition of $\sigma_\pi$  will contain each primitive cycle as a (possibly non-contiguous) subsequence. This means, the repetition, and thus $\sigma_\pi$, must have a strictly negative $\bx$-weighted balance.
For the converse, we observe that $\sigma_\pi$ can be decomposed into primitive cycles. Thus, if $\sigma_\pi$ has strictly negative $\bx$-weighted balance, then so must at least one of its constituent primitive cycles.

Applying Farkas' Lemma to $\bA_\pi\bx \leq \bb$ either yields a solution $\bx \in \Q_{\geq0}^n$ or a vector $\by \in \Q_{\geq0}^{m}$ with $\by^\top \bA_\pi \geq \bzero$ and $\by^\top \bb < 0$.
In both cases we assume wlog.\ that the given vector has entries in $\N$, as we can always multiply with the lcm of the denominators.

Suppose we have a solution $\bx$.
We claim that then $L_\pi(\cA)\subseteq S_{\bx,k}$, where
$k = |Q_\pi| \cdot h$ and $h$ is the maximal length of a transition label of $\cA$.
This is because $\bx$ weights primitive cycles non-positively, and $k$ is chosen such that for any infix $v$ of a word in $L_\pi(\cA)$, if $|v| > k$, then $v$'s associated transition sequence has to contain a primitive cycle. Thus, infixes at almost all start positions of a word in $L_\pi(\cA)$ must have $\bx$-weighted balance $\le k$.

If we obtain a vector $\by = (y_1,\ldots,y_m)$, then we can view it as a selection of rows in the matrix $\bA_\pi$, where the $j$th row is being selected $y_j$ many times.
Since each row corresponds to a cycle, this is also a selection of cycles.
Then by $\by^\top \bb < 0$ we selected $\sigma_\pi$, where we can insert the other selected cycles.
By $\by^\top \bA_\pi \geq \bzero$ this forms a cycle $\sigma'$ as required, with non-negative letter balance for all letter pairs.
A detailed proof can be found in \cref{appendix-basic-separators}.

Here, we used a system of linear inequalities $\bA_\pi\bx \leq \bb$, which was solely dependent on $\cA$ and $\pi$.
We reasoned that if this system has a solution, then Condition~(\labelcref{post-separator-exists}) has to hold.
This is a fact that we want to refer to in a later proof, and therefore we formalize it here.
\begin{corollary}\label{solution-implies-separator}
  If $\cA$ is a B\"{u}chi automaton with a profile $\pi$ for which there
  is an $\bx\in\N^n$ with $\bA_\pi\bx\le\bb$, then $L_\pi(\cA)\subseteq
  S_{\bx,k}$ for some $k\in\N$.
\end{corollary}
With \cref{make-pumpable} and \cref{lem-basic-post-separators}, we can now show
\cref{thm-basic-separators}. Suppose $R=L(\cA)$ for some B\"{u}chi automaton
$\cA$. First, applying \cref{make-pumpable} with $d=0$ yields a B\"{u}chi
automaton $\pump{\cA}$ such that $L(\cA)\subseteq L(\pump{\cA})\cup P_{\ell}$
for some $\ell\in\N$ and $L(\pump{\cA})\cap D_n=\emptyset$. Therefore, it
suffices to show that $L(\pump{\cA})$ is included in a finite union of
languages $S_{\bx,k}$. Suppose not. Then the set $L(\pump{\cA})$ decomposes
into the sets $L_\pi(\pump{\cA})$ for $\pi\in\prof{\pump{\cA}}$. By
\cref{lem-basic-post-separators}, we know that for some $\pi$,
Condition~(\labelcref{post-separator-nonex-cycle}) must hold: Otherwise, each
$L_\pi(\pump{\cA})$ would be included in some $S_{\bx,k}$. But if
(\labelcref{post-separator-nonex-cycle}) holds for $\pi$, then there is a cycle
$\sigma'$ in $\pump{\cA}$ that contains $\pi$ (and thus visits a final state)
and reads a word $v$ with $\varphi(v)\ge\bzero$. Now for some finite prefix
$u$, the word $uv^\omega$ belongs to $L(\pump{\cA})$. Since
$\varphi(v)\ge\bzero$, there is some lower bound $B\in\Z$ such that for each
$i\in[1,n]$ and every prefix $p$ of $uv^\omega$, we have $\varphi_i(p)\ge B$.
Finally, since $L(\pump{\cA})$ is pumpable, we can exchange a prefix in
$w=uv^\omega$ to obtain another word $w'\in L(\pump{\cA})$ where every prefix
$p$ has $\varphi(p)\ge \bzero$. Hence $w'\in D_n$ and thus
$L(\pump{\cA})\cap D_n\ne\emptyset$, a contradiction.

\section{Deciding Regular Separability}\label{sec:decidability}
We now present the algorithm to decide, given a B\"{u}chi VASS $\cV$ whether $\regsep{L(\cV)}{D_n}$. 
We first employ \cref{make-pumpable},  because
for pumpable languages we only have to deal with one type of basic separators.
The next step is to generalize the notion
of profiles from B\"{u}chi automata to B\"{u}chi VASS.
 Recall that for a
sequence $\chi$ of transitions in $\cV$, $\inteff(\chi)$ denotes its effect on
the counters of $\cV$.  If $\chi$ is a transition sequence in $\KM{\cV}$, then
$\chi$ is labeled with a transition sequence of $\cV$, so we define
$\inteff(\chi)$ accordingly. Since we consider B\"{u}chi VASS with input
alphabet $\Sigma_n$, we write $\exteff(\chi)$ for the image of the input word
under $\exteff$. Again, this notation is used for transition sequences in
$\KM{\cV}$. We also write $\eff(\chi)=(\inteff(\chi),\exteff(\chi))$.
\begin{definition}
	Let $\cV$ be a B\"{u}chi VASS. A \emph{profile for $\cV$} is a set
	$\pi$ of edges in $\KM{\cV}$ for which there exists a cycle $\sigma$ in
	$\KM{\cV}$ such that (i)~$\sigma$ contains exactly the edges in $\pi$,
	(ii)~$\sigma$ starts (and ends) in a final state, and
	(iii)~$\inteff(\sigma)\ge \bzero$.
\end{definition}
Clearly, every B\"{u}chi VASS has a finite set of profiles, which we denote by
$\prof{\cV}$. Moreover, $\prof{\cV}$ can be constructed effectively:
Given a set of edges, a simple reduction to
checking unboundedness of a counter
can be used to check if it is a profile.
Furthermore, to every run $\rho$ of $\cV$, we can associate a profile: The run
$\rho$ must have a corresponding run in $\KM{\cV}$, which has a finite set
$\prof{\rho}$ of edges that are used infinitely often. Thus, $\rho$ decomposes
as $\rho_0\rho_1$ such that $\rho_1$ only contains edges from $\pi$. Then,
$\rho_1$ decomposes into $\sigma_1\sigma_2\cdots$ such that each $\sigma_i$
uses every edge from $\Pi(\rho)$ at least once and starts (and ends) in a final
state. Since $\le$ is a well-quasi ordering on $\N^n$, there are $r<s$
such that $\inteff(\sigma_r\cdots \sigma_s)\ge \bzero$. Thus, $\sigma=\sigma_r\cdots
\sigma_s$ is our desired transition sequence showing that $\prof{\rho}$ is a
profile. For each $\pi\in\prof{\cV}$, we denote by $L_\pi(\cV)$ the set of
all words accepted by runs $\rho$ of $\cV$ for which $\prof{\rho}=\pi$. Then clearly:
\begin{lemma}
$L(\cV)=\bigcup_{\pi\in\prof{\cV}} L_\pi(\cV)$.
\end{lemma}

\subparagraph{A system of inequalities for each profile} 
Our next step is to associate with each profile $\pi\in\prof{\cV}$ a system of linear inequalities.
We need some terminology.  %
A \emph{$\pi$-cycle} is a cycle $\sigma$ in $\KM{\cV}$ that only contains edges
in $\pi$. If in addition, $\sigma$ visits each state of $\KM{\cV}$ at most
once, except for the initial state, which is visited twice, then $\sigma$ is a
\emph{primitive $\pi$-cycle}.  Clearly, a primitive $\pi$-cycle has length
$\le|\pi|$.
Moreover, from every $\pi$-cycle $\sigma$, one can successively cut
out primitive $\pi$-cycles until it is empty.  Therefore, if
$\tau_1,\ldots,\tau_m$ are the primitive $\pi$-cycles of $\KM{\cV}$, then there are numbers
$r_1,\ldots,r_m\in\N$ such that
$\eff(\sigma)=r_1\cdot\eff(\tau_1)+\cdots+r_m\cdot\eff(\tau_m)$.  We call
$\sigma$ a \emph{complete $\pi$-cycle} if this holds for some
$r_1,\ldots,r_m\ge 1$. 
Observe that if $\pi$ is a profile, then this is always witnessed by a complete
$\pi$-cycle: Take any cycle $\sigma$ witnessing that $\pi$ is a profile.
Then $\sigma^{|\pi|}$ contains each primitive $\pi$-cycle as a subsequence.
Hence, the cycle $\sigma^{m\cdot|\pi|}$ is complete: We can carry out the
cutting in each factor $\sigma^{|\pi|}$ so as to cut some $\tau_i$ at least
once. Moreover, $\sigma^{m\cdot|\pi|}$ still witnesses that $\pi$ is a profile,
since $\inteff(\sigma^{m\cdot|\pi|})=m\cdot|\pi|\cdot\inteff(\sigma)\ge \bzero$.

Let us now construct the system of inequalities associated with $\pi$. Let
$\sigma$ be a complete $\pi$-cycle witnessing that $\pi$ is a profile and let
$\tau_1,\ldots,\tau_m$ be the primitive $\pi$-cycles.  Let
$\bA_\pi\in\Z^{(m+1)\times n}$ be the matrix with rows
$\exteff(\tau_1),\ldots,\exteff(\tau_m),\exteff(\sigma)$, and let
$\bb\in\Z^{m+1}$ be the column vector $(0,\ldots,0,-1)$. Then clearly,
$\bA_\pi\bx\le\bb$ is equivalent to $\langle
\bx,\exteff(\sigma)\rangle<0$ and $\langle \bx,\exteff(\tau)\rangle\le 0$ for
each primitive $\pi$-cycle $\tau$.
\subparagraph{Inseparability flowers} An \emph{inseparability flower} is a 
structure in the Karp-Miller graph\vspace{0.05cm}
\begin{minipage}{5.5cm}
 $\KM{\cV}$ as depicted to the right. 
 It consists of a final state $q$ and 
three cycles $\alpha,\beta,\gamma$
that all start in $q$ and that meet the given conditions.
\end{minipage}\hspace{0.4cm}
\begin{minipage}{5cm}
\begin{tikzpicture}[initial text={}]
\node[state,initial] (q0) {};
\node[state,accepting,right=1.5cm of q0] (q1) {$q$};
\path[->] 
(q0) edge [dashed] (q1)
(q1) edge [loop above, out=130, in=50, looseness=6,densely dashed] node[align=center, below]{$\alpha$} (q1)
(q1) edge [loop right, out=40, in=-40, looseness=6,densely dashed] node[align=center, left]{$\beta$} (q1)
(q1) edge [loop below, out=-50, in=-130, looseness=6,densely dashed] node[align=center, above]{$\gamma$} (q1)
;
\node[right=1.3cm of q1, align=left] (conditions) {$\inteff(\alpha\beta\gamma)\ge\bzero$ \\ $\exteff(\alpha\beta)\ge\bzero$ \\ $\exteff(\alpha\beta\gamma)\in\Q\cdot\exteff(\alpha)$};
\end{tikzpicture}
\end{minipage}

Let us give some intuition on why such a flower is the relevant structure to look for.
True to its name, an inseparability flower guarantees the existence of an inseparability witness, i.e.\ a family of words accepted by the pumpable B\"{u}chi VASS $\cV$ that escape every basic separator $S_{\bx,k}$.
Such a family of words therefore needs an accepting run for each member, and the three conditions of the flower provide such runs:
The first condition ensures that the three cycles actually correspond to a transition sequence enabled in $\cV$.
The second condition guarantees that for every $\bx\in\N^n$, the $\bx$-weighted letter balance of $\alpha$ or of $\beta$ is positive; unless they are both zero, in which case
the third condition ensures that $\alpha\beta\gamma$ has $\bx$-weighted balance zero.
This allows us to construct, for each $k$, a run that escapes $S_{\bx,k}$ for all $\bx$:
By sufficiently repeating each cycle $\alpha$, $\beta$, and $\gamma$, we obtain a run that for each $\bx\in\N^n$,
will either (i)~have infixes with $\bx$-weighted balance $>k$, or (ii)~attain some $\bx$-weighted balance infinitely often. Each of these properties rules out membership in $S_{\bx,k}$.
\cref{flower-implies-insep} proves this formally.

\begin{theorem}\label{dec}
	Let $\cV$ be a B\"{u}chi VASS such that $L(\cV)$ is pumpable.
	Then the following are equivalent:
	\begin{myenum}
		\myitem\label{dec-not-separable} $\notregsep{L(\cV)}{D_n}$.
		\myitem\label{dec-no-solution} There is a profile $\pi\in\prof{\cV}$ such that the system $\bA_\pi\bx\le\bb$ has no solution $\bx\in\N^n$.
		\myitem\label{dec-inseparability-flower} There exists an inseparability flower in $\KM{\cV}$.
	\end{myenum}
\end{theorem}
\subparagraph{The decision procedure}
Before we prove \cref{dec}, let us see how to use it to decide 
separability. 
Given B\"{u}chi VASS $\cV_1$ and $\cV_2$, we 
can compute $\cV$ %
so that $\regsep{L(\cV_1)}{L(\cV_2)}$ if
and only if $\regsep{L(\cV)}{D_n}$, by \cref{one-language-fixed}. 
Then \cref{make-pumpable} tells us that
$L(\pump{\cV})$ is pumpable and we have $\regsep{L(\cV)}{D_n}$ if and only if
$\regsep{L(\pump{\cV})}{D_n}$. Finally, by \cref{dec}, we can check whether
$\regsep{L(\pump{\cV})}{D_n}$ by checking %
the systems
$\bA_{\pi}\bx\le\bb$ for satisfiability: %
If there is
a solution for every $\pi\in\prof{\pump{\cV}}$, then we have separability; otherwise, we have
inseparability.
Since the systems $\bA_{\pi}\bx\le\bb$ are constructed directly from $\KM{\pump{\cV}}$,
we need to explicitly construct the latter.
Therefore our procedure may take Ackermann time, because Karp-Miller graphs can be Ackermann large~\cite[Theorem 2]{mayr1981complexity}.

\begin{example}
  Consider the instance of regular separability where our two inputs are the B\"{u}chi VASS $\cV$ found in \cref{Figure:examples}(left), and another B\"{u}chi VASS accepting the language $D_1$.
  Since we are already in the case of wanting to decide $\regsep{L(\cV)}{D_1}$, we can skip the first step of applying \cref{one-language-fixed}.
  The second step is to apply \cref{make-pumpable} and construct $\pump{\cV}$, which we have already done for this case in \cref{Figure:pumpability}(right).
  
  \begin{figure}[t]
  \begin{center}
    \scalebox{0.85}{
      \begin{tikzpicture}[initial text={},baseline]
        \node (q0) {$((q_0,\omega,0),1)$};
        \node[right=1cm of q0] (q1) {$((q_1,\omega,0),1)$};
        \node[right=1cm of q1] (q2) {$((q_2,\omega,0),1)$};
        \node[right=1cm of q2] (q3) {$((q_2,\omega,\omega),0)$};
        \node[right=1cm of q3] (q4) {$((q_1,\omega,\omega),0)$};
        \node[below=1cm of q1] (p1) {$((q_1,0,0),0)$};
        \node[initial,left=1cm of p1] (p0) {$((q_0,0,0),0)$};
        \node[right=1cm of p1] (p2) {$((q_2,0,0),0)$};
        \node[above=1cm of q1] (r1) {$((q_1,\omega,0),\omega)$};
        \node[left=1cm of r1] (r0) {$((q_0,\omega,0),\omega)$};
        \node[right=1cm of r1] (r2) {$((q_2,\omega,0),\omega)$};
        \node[above=1cm of r2] (s2) {$((q_2,\omega,\omega),\omega)$};
        \node[left=1cm of s2] (s1) {$((q_1,\omega,\omega),\omega)$};
        \path[->] 
          (p0) edge node[left] {$\be_1|\varepsilon$} (q0)
          (p0) edge node[below] {$\bzero|\varepsilon$} (p1)
          (p1) edge node[below] {$\bzero|\varepsilon$} (p2)
          
          (q0) edge node[left] {$\be_1|\varepsilon$} (r0)
          (q0) edge node[below] {$\bzero|\varepsilon$} (q1)
          (q1) edge node[below] {$\bzero|\varepsilon$} (q2)
          (q2) edge node[below] {$-\be_1|a_1$} (q3)
          
          (q3) edge[bend right=10] node[below] {$\bzero|\bar{a}_1$} (q4)
          (q4) edge[bend right=10] node[above] {$\bzero|\varepsilon$} (q3)
          (q3) edge[bend right=30] node[right] {$\be_1|\bar{a}_1$} (s2)
          
          (r0) edge [loop above] node {$\be_1|\varepsilon$} (r0)
          (r0) edge node[below] {$\bzero|\varepsilon$} (r1)
          (r1) edge node[below] {$\bzero|\varepsilon$} (r2)
          (r2) edge node[right] {$-\be_1|a_1$} (s2)
          
          (s1) edge[bend left=10] node[above] {$\bzero|\varepsilon$} (s2)
          (s2) edge[bend left=10] node[below] {$\bzero|\bar{a}_1$} (s1)
          (s2) edge[loop above] node {$\begin{aligned}
              -\be_1|a_1 \\
              \be_1|\bar{a}_1
            \end{aligned}$} (s2)
        ;
      \end{tikzpicture}
    }%
  \end{center}
  \vspace{-0.5cm}
  \caption{The Karp-Miller graph $\KM{\pump{\cV}}$ of the B\"{u}chi VASS $\pump{\cV}$ from \cref{Figure:pumpability}(left).
  Here we did not mark the final states to reduce visual clutter; every state that includes $q_1$ is considered final.
  For similar reasons, we also only labelled the edges of the graph with letters and counter effects.
  The proper edge labels would be full transitions of $\pump{\cV}$, including source and target state.
  \label{Figure:KMpump}}
\end{figure}
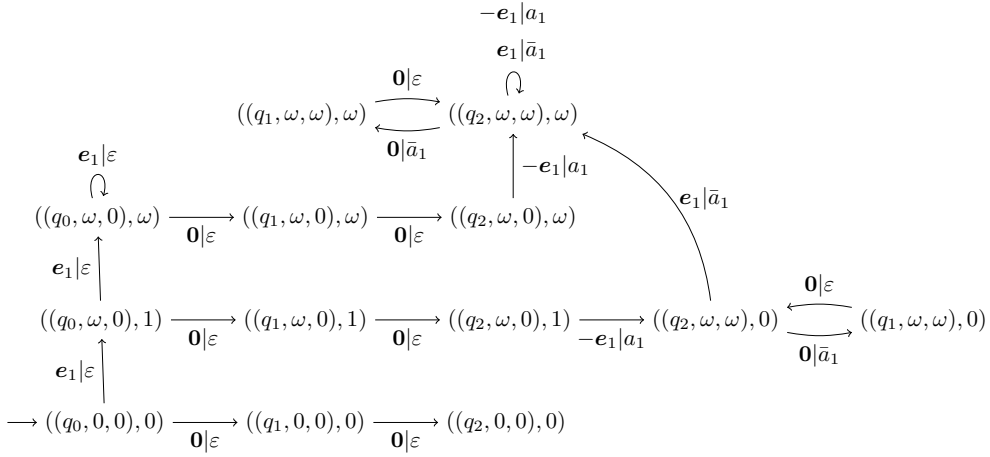
  
  Now we have to construct $\KM{\pump{\cV}}$, which can be found in \cref{Figure:KMpump}.
  There are two relevant parts of $\KM{\pump{\cV}}$, where we can find cycles involving a final state:
  (1) the part on the right, where the state tuples contain $\omega$ twice and the counter value is $0$, and (2) the part at the top with triple $\omega$s.
  In the following we will only write down the states, as the counter values and the other contents of the state tuples will be clear from context.
  
  For part (1), the B\"{u}chi VASS $\pump{\cV}$ has only a single profile $\pi_1$ containing only the two edges between $q_1$ and $q_2$.
  Since each $\pi_1$-cycle $\sigma$ only consists of repetitions of the primitive cycle $q_1 \xrightarrow{\bzero|\varepsilon} q_2 \xrightarrow{\bzero|\bar{a}_1} q_1$, we have $\varphi(\sigma) < 0$.
  Therefore the system $\bA_{\pi_1}\bx\le\bb$ trivially has a solution $\bx = 1$.
  
  Regarding part (2), $\pump{\cV}$ has exactly two more profiles: profile $\pi_2$ containing only the two edges between $q_1$ and $q_2$, and profile $\pi_3$, which additionally contains the two loop edges on $q_2$.
  The cycles of $\pi_2$ look almost exactly like the cycles of $\pi_1$ with only the counter values of the nodes in the graph being different.
  Thus, the system $\bA_{\pi_2}\bx\le\bb$ is the exact same system as $\bA_{\pi_1}\bx\le\bb$ and also trivially has a solution $\bx = 1$.
  
  For $\pi_3$, we have as primitive cycles both the loop edges on $q_2$ as well as the primitive cycle of $\pi_2$.
  To obtain a complete $\pi_3$-cycle, we simply insert both loops into the $\pi_2$-cycle at $q_2$ forming the cycle $\sigma = q_1 \xrightarrow{\bzero|\varepsilon} q_2 \xrightarrow{-\be_1|a_1} q_2 \xrightarrow{\be_1|\bar{a}_1} q_2 \xrightarrow{\bzero|\bar{a}_1} q_1$.
  Since $\sigma$ contains all primitive cycles exactly once without overlap, it is automatically complete.
  We also have $\delta(\sigma) = 0$, meaning $\sigma$ is a cycle witnessing $\pi_3$ as a profile.
  Thus these cycles lead to the following system of inequalities $\bA_{\pi_3}\bx\le\bb$:
  \begin{alignat*}{2}
    1\cdot x_1 &\leq 0 &\qquad &\text{loop 1}\\
    -1\cdot x_1 &\leq 0 &\qquad &\text{loop 2} \\
    -1\cdot x_1 &\leq 0 &\qquad &\pi_2\text{-cycle} \\
    -1\cdot x_1 &\leq -1 &\qquad &\text{complete }\pi_3\text{-cycle}
  \end{alignat*}
  Clearly this system has no solution; the first and last inequality are contradictory.
  Therefore we conclude regular inseparability for $L(\cV)$ and $D_1$.
  
  While not part of the decision procedure, for an inseparable instance of the problem as we have here, we can also find an inseparability flower in $\KM{\pump{\cV}}$.
  In this case we have
  $\alpha = q_1 \xrightarrow{\bzero|\varepsilon} q_2 \xrightarrow{\bzero|\bar{a}_1} q_1$,
  $\beta = q_1 \xrightarrow{\bzero|\varepsilon} q_2 \xrightarrow{-\be_1|a_1} q_2 \xrightarrow{-\be_1|a_1} q_2 \xrightarrow{\bzero|\bar{a}_1} q_1$, and
  $\gamma= q_1 \xrightarrow{\bzero|\varepsilon} q_2 \xrightarrow{\be_1|\bar{a}_1} q_2 \xrightarrow{\be_1|\bar{a}_1} q_2 \xrightarrow{\bzero|\bar{a}_1} q_1$.
  This selection of cycles meets all the requirements of a flower:
  $\delta(\alpha\beta\gamma) = 0$, $\varphi(\alpha\beta) = 0$, and $\varphi(\alpha\beta\gamma) = -3 = 3 \cdot \varphi(\alpha)$.
\end{example}

\subparagraph{Inseparability flowers disprove separability}
The remainder of this section is devoted to proving \cref{dec}.
The implication \impl{dec-not-separable}{dec-no-solution} follows by
applying \cref{solution-implies-separator} to $\KM{\cV}$, viewed as a B\"{u}chi automaton; see
\cref{x-yields-separator}.  For
\impl{dec-inseparability-flower}{dec-not-separable},  we employ
\cref{prefix-independent-limsep}:
\begin{proposition}\label{flower-implies-insep}
	If $L(\cV)$ is pumpable and $\KM{\cV}$ has an
	insep.\ flower, then $\notregsep{L(\cV)}{D_n}$.
\end{proposition}
\begin{proof}
	Suppose there is an inseparability flower $\alpha,\beta,\gamma$ in $\KM{\cV}$ and also $\regsep{L(\cV)}{D_n}$. By
	\cref{prefix-independent-limsep}, there is a $k\in\N$ and a finite set
	$X\subseteq\N^n$ such that $L(\cV)\subseteq\bigcup_{\bx\in X}
	S_{\bx,k}$.
	We claim that for every $\bx\in\N^n$, at least one of the following holds:
	\begin{align}
		&\langle \bx,\exteff(\alpha)\rangle>0,  & &\langle\bx,\exteff(\beta)\rangle>0, & &\text{or}~\langle \bx,\exteff(\alpha\beta\gamma)\rangle=0. \label{rule-out-u}
	\end{align}
	Indeed, if $\langle \bx,\exteff(\alpha)\rangle\le 0$ and
	$\langle\bx,\exteff(\beta)\rangle\le 0$, then
	$\exteff(\alpha\beta)\ge\bzero$ implies that
	$\langle\bx,\exteff(\alpha)\rangle=\langle\bx,\exteff(\beta)\rangle=0$.
	Since $\exteff(\alpha\beta\gamma)=N\cdot\exteff(\alpha)$
	for some $N\in\Q$, we have $\exteff(\alpha\beta\gamma)=\bzero$. This proves the claim.
	Because of \eqref{rule-out-u}, the sequence $\alpha^{k+1}\beta^{k+1}\gamma^{k+1}$ either has an infix $\chi$ with $\langle\bx,\exteff(\chi)\rangle>k$ or we have $\langle\bx,\exteff(\alpha^{k+1}\beta^{k+1}\gamma^{k+1})\rangle=0$.
	Since $\inteff(\alpha^{k+1}\beta^{k+1}\gamma^{k+1})\ge\bzero$, there is a run $\rho$ such that $\rho\alpha^{k+1}\beta^{k+1}\gamma^{k+1}$ is a run in $\cV$.
	Hence, $\rho(\alpha^{k+1}\beta^{k+1}\gamma^{k+1})^\omega$ is a run in $\cV$ whose word cannot belong to $S_{\bx,k}$ for any $\bx\in\N^n$, contradicting $L(\cV)\subseteq\bigcup_{\bx\in X} S_{\bx,k}$.\qedhere
\end{proof}

\subparagraph{Constructing inseparability flowers} 
It remains to show the implication
\impl{dec-no-solution}{dec-inseparability-flower}.  Suppose there is a profile
$\pi\in\prof{\cV}$ whose associated system of inequalities $\bA_\pi\bx\le\bb$ is unsatisfiable. 
By Farkas' Lemma, there exists a $\by\in\N^{m+1}$ such that
$\by^\top\bA_\pi\ge\bzero$ and $\by^\top\bb<0$. From this vector $\by$,
we now construct an inseparability flower in $\KM{\cV}$.

Let $\sigma$ be the complete $\pi$-cycle in $\KM{\cV}$ that was chosen to
construct $\bA_\pi$. 
Let $\tau_1,\ldots,\tau_m$ be the primitive $\pi$-cycles.  
Since $\sigma$ is complete, there is a vector
$\br=(r_1,\ldots,r_m)\in\N^m$ so that $r_1,\ldots,r_m\ge 1$ and
$\eff(\sigma)=r_1\cdot\eff(\tau_1)+\cdots+r_m\cdot\eff(\tau_m)$. 
Moreover, since $\sigma$ contains every edge of $\pi$, we can wlog.\ write
$\sigma=\sigma_0\cdots \sigma_m$ such that between $\sigma_{i-1}$ and
$\sigma_{i}$, $\sigma$ arrives in the initial state of $\tau_i$.
The decomposition allows us to insert further repetitions of the primitive cycles.
For $\bz=(z_1,\ldots,z_m)\in\N^m$ with
$\bz\ge\br$, we define $\sigma^{\bz}$ as
$\sigma_0\tau_1^{z_1-r_1}\sigma_1\cdots \tau_m^{z_m-r_m}\sigma_m$. 
Then 
$\eff(\sigma^{\bz})=z_1\cdot\eff(\tau_1)+\cdots+z_m\cdot\eff(\tau_m)$. 
In particular, for $\bs,\bt\ge\br$, we have $\eff(\sigma^{\bs}\sigma^{\bt})=\eff(\sigma^{\bs+\bt})$.

\newcommand{\trans}{\mathsf{trans}}
Recall that every transition in a Karp-Miller graph is labeled by a VASS transition, and so every transition sequence $\chi$ in $\KM{\cV}$ is labeled by a transition sequence in $\cV$, which we denote by $\trans(\chi)$. 
We now define the transition sequences $\alpha$, $\beta$, and $\gamma$ as $\trans(\sigma^{\bz})$ for suitable vectors $\bz$. 
For $\alpha$, we take $\trans(\sigma)$, the transitions labeling the complete $\pi$-cycle. 
Observe that $\sigma= \sigma^{\br}$. 
We proceed to define $\beta=\trans(\sigma^{\bs})$ and $\gamma=\trans(\sigma^{\bt})$. 
The choice of the vectors $\bs$ and $\bt$ has to meet the requirements on an inseparability flower: $\exteff(\alpha\beta)\ge \bzero$, $\inteff(\alpha\beta\gamma)\ge\bzero$, and $\exteff(\alpha\beta\gamma)\in\Q\cdot\exteff(\alpha)$.

\subparagraph{Step I: Building $\boldsymbol{\beta}$.} 
We will define $\bs$ so that
$\exteff(\alpha\beta)=\exteff(\sigma^{\br}\sigma^{\bs})=\exteff(\sigma^{\br+\bs})\ge\bzero$.
The remaining two requirements (i.e.\ $\delta(\alpha\beta\gamma)\ge\bzero$ and $\varphi(\alpha\beta\gamma)\in\Q\cdot\varphi(\alpha)$) will be ensured with an appropriate choice of~$\bt$ in Step II. Let us now describe how to pick $\bs$.
Recall that $\by$ is the vector from the application of Farkas' Lemma. 
It can be understood as assigning a repetition count $y_i$ to every primitive cycle $\tau_i$ in the profile and a repetition count $y_{m+1}$ to the complete $\pi$-cycle $\sigma$. 
Since $\by^\top\bA_\pi\ge\bzero$, and since our goal is to make  $\exteff(\alpha\beta)$ non-negative, we will use $\by$ to construct a vector $\hat{\by}=(\hat{y}_1,\ldots,\hat{y}_m)\in\N^m$ so that $\exteff(\sigma^{\hat{\by}})=\by^\top\bA_\pi$. 
The right definition is $\hat{y}_i:=y_i+y_{m+1}\cdot r_i$ for $i\in[1,m]$, because
\begin{align*} 
\by^\top\bA_\pi = \sum_{i=1}^m y_i\cdot\exteff(\tau_i)+y_{m+1}\cdot\exteff(\sigma) = 
\sum_{i=1}^m (y_i+y_{m+1}r_i)\exteff(\tau_i)=\exteff(\sigma^{\hat{\by}}).
\end{align*}
We now choose $M\in\N$ such that $\bs=M\cdot\hat{\by}-\br\ge \br$. 
This is possible since
all entries in $\hat{\by}$ are positive,
due to $y_{m+1} > 0$ by $\by^\top\bb<0$, and $r_i > 0$ for all $i$ by definition.
Then we have
	$\exteff(\alpha\beta)=\exteff(\sigma^{\br}\sigma^{\bs})=\exteff(\sigma^{\br+\bs})=\exteff(\sigma^{M\cdot\hat{\by}})=M\cdot\exteff(\sigma^{\hat{\by}})\ge\bzero$. 

\subparagraph{Step II: Building $\boldsymbol{\gamma}$.}
It remains to define $\bt$ so that $\gamma=\trans(\sigma^{\bt})$ satisfies
$\inteff(\alpha\beta\gamma)=\inteff(\sigma^{\br+\bs+\bt})\geq \bzero$ and
$\exteff(\alpha\beta\gamma)\in\Q\cdot\exteff(\alpha)$. 
The idea is to choose $\bt$ %
so that $\br+\bs+\bt$ is a positive multiple of $\br$. 
Such a choice is possible, because $\br$ has positive entries everywhere:
We pick $N\in\N$ such that $\bt:=N\cdot\br-\bs-\br\ge\br$.
Then indeed
	$\inteff(\alpha\beta\gamma)=\inteff(\sigma^{\br+\bs+\bt})=\inteff(\sigma^{N\cdot\br})=N\cdot \inteff(\sigma^{\br})=N\cdot \inteff(\sigma)\ge\bzero$ and
	$\exteff(\alpha\beta\gamma)=\exteff(\sigma^{\br+\bs+\bt})=\exteff(\sigma^{N\cdot\br})=N\cdot\exteff(\sigma^{\br})=N\cdot\exteff(\alpha)$.

\section{One-dimensional B\"{u}chi VASS}\label{sec:one-dim}
Our second contribution is the precise complexity of separability for the $1$-dimensional case.
\begin{restatable}{theorem}{oneDimLowerBound}\label{one-dim-complexity}
  Regular separability for $1$-dimensional B\"{u}chi VASS with binary encoded
	updates is $\PSPACE$-complete.
\end{restatable}
For the lower bound, we use a simple reduction from the disjointness
problem $L_1\cap L_2\overset{?}{=}\emptyset$ for finite-word languages of
$1$-dim.\ VASS~\cite{DBLP:journals/iandc/FearnleyJ15}. However, we also show
that separability is $\PSPACE$-hard even if the input languages are promised to be
disjoint. See \cref{appendix-one-dim-hardness}.

For the upper bound, we rely on the results in \cref{sec:decidability}, but need a modification. 
There, to simplify the exposition, we first make the input language
pumpable, which may incur an Ackermannian blowup. 
A closer look at the results, however, reveals that we can also check separability
directly on the Karp-Miller graph of $\bar{\cV}$ as defined in Section~\ref{Section:Outline}.
\begin{restatable}{proposition}{flowerInVBarSufficient}\label{flower-in-Vbar-sufficient}
Let $\cV$ be a B\"{u}chi VASS with $L(\cV)\subseteq\Sigma_n^\omega$. Then
$\notregsep{L(\cV)}{D_n}$ if and only if $\KM{\bar{\cV}}$ has an inseparability
flower.
\end{restatable}

\cref{flower-in-Vbar-sufficient} allows us to phrase inseparability as the existence of a run in
$\bar{\cV}$ that satisfies certain constraints. 
Recall that if $\cV$ is $1$-dimensional and over $\Sigma_1$, then $\bar{\cV}$ has two counters the second of which tracks the letter balance.
\begin{corollary}\label{witness-run-2dim}
	Let $\cV$ be a $1$-dimensional B\"{u}chi VASS with
	$L(\cV)\subseteq\Sigma_1^\omega$ and $L(\cV)\cap D_1=\emptyset$. Then
	$\notregsep{L(\cV)}{D_1}$ if and only if there exist states $p,q,r$ with 
	$r$ final, and a run in $\bar{\cV}$ as follows:
	\begin{minipage}{4cm}
	\footnotesize
	\begin{align*}
		(q_0,0,0)&\autsteps\overbrace{(p,x_1,y_1)\autsteps(p,x_2,y_2)}^{\sigma_1}
		\autsteps\overbrace{(q,x_3,y_3)\autsteps(q,x_4,y_4)}^{\sigma_2} \\
		&\autsteps\lefteqn{\overbrace{\phantom{(r,x_5,y_5)\autsteps(r,x_6,y_6)}}^{\alpha}}(r,x_5,y_5)\autsteps\underbrace{(r,x_6,y_6)\autsteps \lefteqn{\overbrace{\phantom{r,x_7,y_7)\autsteps(r,x_8,y_8)}}^{\gamma}}(r,x_7,y_7)}_{\beta}\autsteps (r,x_8,y_8) 
	\end{align*}
	\end{minipage}
	\qquad
	\begin{minipage}{7cm}
	\footnotesize
	\begin{enumerate}[(1)]
		\item\label{witness-run-make-omega} $y_3<y_4$ and also (a)~$x_3\le x_4$ \\
		or (b)~$x_1<x_2$ and $y_1\le y_2$
		\item\label{witness-run-exteff} $y_5\le y_7$
		\item\label{witness-run-inteff} $x_5\le x_8$
		\item\label{witness-run-multiple} if $y_5=y_6$, then $y_5=y_8$.\vspace{0.1cm}
	\end{enumerate}
	\end{minipage}
\end{corollary}

Observe that an inseparability flower in $\KM{\bar{\cV}}$
must carry $\omega$ in the second coordinate, meaning the letter balance is unbounded.  Otherwise, it would yield an
accepting run of~$\bar{\cV}$, which cannot exist because $L(\cV)\cap
D_1=\emptyset$. 
If the flower has $\omega$ in the second
coordinate, we can construct a \emph{finite} run as above. 
The
cycles $\sigma_1$ and $\sigma_2$ plus  Condition~\labelcref{witness-run-make-omega} ensure that indeed the second coordinate
becomes $\omega$. %
Condition~\labelcref{witness-run-exteff} is
$\exteff(\alpha\beta)\ge\bzero$. 
Condition~\labelcref{witness-run-inteff} says
$\inteff(\alpha\beta\gamma)\ge\bzero$. Finally, to express
$\exteff(\alpha\beta\gamma)\in\Q\cdot\exteff(\alpha)$, note that for
integers $a\in\Q\cdot b$ iff $b=0$ implies
$a=0$. 
Condition~\labelcref{witness-run-multiple} expresses that
$y_6-y_5=0$ implies $y_8-y_5=0$.

In order to apply \cref{witness-run-2dim} for deciding
$\regsep{L(\cV_1)}{L(\cV_2)}$ for $1$-dim.\ B\"{u}chi VASS $\cV_1,\cV_2$ with
binary counter updates, we would like to follow the approach for the general case
and use \cref{one-language-fixed} to first construct $\cV$ so
that $\regsep{L(\cV_1)}{L(\cV_2)}$ if and only if $\regsep{L(\cV)}{D_1}$. 
From $\cV$, we would then construct the $2$-dimensional B\"uchi VASS $\bar{\cV}$ that tracks the letter balance, and on $\bar{\cV}$ we would then check the conditions of Corollary~\ref{witness-run-2dim}.
The problem is that, under binary updates, the intermediary $\cV$ may become exponentially large.
We use the fact that also $\bar{\cV}$ has binary counters available.
This allows us to directly construct a compact variant of $\bar{\cV}$:
\begin{restatable}{lemma}{compactVBar}\label{Lemma:CompactVBar}
	Given $1$-dim.\ B\"uchi VASS $\cV_1,\cV_2$ with binary updates, there is a a $1$-dim.\ B\"uchi VASS $\cV$
  with $L(\cV_1)\cap L(\cV_2)=\emptyset$ iff $L(\cV)\cap D_1=\emptyset$, $\regsep{L(\cV_1)}{L(\cV_2)}$ iff $\regsep{L(\cV)}{D_1}$,
  and we can construct in time polynomial in $|\cV_1|+|\cV_2|$ the
  $2$-dim.\ B\"uchi VASS $\bar{\cV}$ (binary updates).
\end{restatable}

\subparagraph{Detecting constrained runs in $2$-VASS}
It remains to check for the existence of runs in $\bar{\cV}$ as described in
\cref{witness-run-2dim}, and to check whether $L(\cV_1)\cap L(\cV_2)=\emptyset$.
Both of these problems reduce to what we call the
\emph{constrained runs problem for $2$-VASS}.
Recall that \emph{Presburger arithmetic} is the first-order theory of $(\N,+,<,0,1)$.
We will use the existential fragment to express conditions on counter values
of VASS like the ones from \cref{witness-run-2dim}.
The \emph{constrained runs problem} is the following:
\begin{description}
	\item[Given] A $2$-dim.\ VASS $\cV$ (with updates encoded in binary), a
		number $m\in\N$, states $q_1,\ldots,q_m$ in $\cV$, a
		quantifier-free Presburger formula
		$\psi(x_1,y_1,\ldots,x_m,y_m)$, and $s,t\in[1,m]$, $s\le t$.
	\item[Question] Does there exist a run
		$(q_0,0,0)\autsteps(q_1,x_1,y_1)\autsteps\cdots\autsteps(q_m,x_m,y_m)$
		that visits a final state between $(q_s,x_s,y_s)$ and
		$(q_t,x_t,y_t)$ and satisfies
		$\psi(x_1,y_1,\ldots,x_m,y_m)$?
\end{description}
\Cref{Lemma:CompactVBar} and \Cref{witness-run-2dim} imply that if $L(\cV_1)\cap
L(\cV_2)=\emptyset$, then $\regsep{L(\cV_1)}{L(\cV_2)}$ reduces to the
constrained runs problem on $\bar{\cV}$. 
Moreover, checking $L(\cV_1)\cap L(\cV_2)=\emptyset$
reduces via a product construction to checking emptiness of a $2$-VASS. 
Such a $2$-VASS has an accepting run iff $(q_0,0,0)\autsteps(q,x,y)\autsteps(q,x',y')$ with $(x,y)\le(x',y')$ and $q$ final. 
Hence, this problem 
also reduces to the constrained runs problem for $2$-VASS. We thus need to show:
\begin{restatable}{proposition}{constrainedRunsInPSPACE}\label{constrained-runs-in-pspace}
	The constrained runs problem for $2$-VASS is solvable in $\PSPACE$.
\end{restatable}
For \cref{constrained-runs-in-pspace}, we show that if there is a constrained run,
then there is one with at most exponential counter values along the way.
For this, we use methods from~\cite{DBLP:journals/jacm/BlondinEFGHLMT21}.

\subparagraph{Complexity in higher dimension} 
We leave open two natural questions: (i)~What is the complexity of regular
separability for B\"{u}chi $d$-VASS, for each $d\ge 2$? (ii)~What is the complexity
of regular separability for B\"{u}chi VASS (where the dimension is part of the
input)?

Given that the regular separability and the disjointness problem usually (but
not always~\cite{Kopczynski16,ThinniyamZetzsche2019a}) coincide regarding
decidability, we expect the complexity of regular separability to be $\PSPACE$
in every fixed dimension $d$ and $\EXPSPACE$ in general. 
The lower bounds follow from
\cref{one-dim-complexity} for fixed $d$ and from~\cite{WSTSRegSep2018} (because disjointness is $\EXPSPACE$-complete~\cite{Esp98a,lipton1976reachability}). However, it is not clear how to show the upper bounds.

The clearest obstacle is that inseparability flowers involve a non-linear
condition: The requirement
$\exteff(\alpha\beta\gamma)\in\Q\cdot\exteff(\alpha)$ is not expressible in
Presburger arithmetic. There are several generic results providing $\EXPSPACE$
upper bounds for detecting particular types of runs in
VASS~\cite{DBLP:journals/jcss/Demri13,DBLP:journals/ijfcs/AtigH11,DBLP:conf/mfcs/BlockeletS11}.
However, the numerical properties directly expressible there are confined to
Presburger arithmetic. 
The only reason we could obtain the $\PSPACE$ upper bound for $d=1$ is that the non-linear condition degenerates into a linear condition in dimension one: It is equivalent to ``$\exteff(\alpha\beta\gamma)=0$ or $\exteff(\alpha)\ne 0$''.

\label{beforebibliography}
\newoutputstream{pages}
\openoutputfile{main.pg}{pages}
\addtostream{pages}{\getpagerefnumber{beforebibliography}}
\closeoutputstream{pages}

\bibliography{references.bib}

\appendix
\section{Proof Details for Overview}\label{appendix-problem}
\subsection{Proof of Part 1 of \cref{Theorem:CounterExamples}}
Here we proof the first part of \cref{Theorem:CounterExamples}, which is the following:
\begin{theorem}\label{counter-example-trivial-regsep}
  There are B\"uchi VASS languages $L_1$, $L_2$ with $L_1\cap L_2= \emptyset$ and $L_1\not\sep L_2$.
\end{theorem}
\begin{proof}\label{proof-counter-example}
  We choose $L_1=L(\cV)$, where $\cV$ is the B\"{u}chi VASS in \cref{Figure:examples}(left), and $L_2=D_1$, the Dyck language.
  We claim that each $w\in L(\cV)$ can be written as $w=uv_1v_2\cdots$ with $\varphi_1(v_\ell)<0$ for every $\ell\in\N$. This clearly implies $L(\cV)\cap D_1=\emptyset$.
  Suppose $w\in L(\cV)$.
  Note that on $q_2$, reading $a_1$ decrements the counter and reading $\bar{a}_1$ increments the counter. Thus, from a configuration $(q_2,x)$, a word $v$ read in $q_2$ can have balance $\varphi_1(v)$ at most $x$. And moreover, if $\varphi_1(v)>0$, then this decreases the counter. Furthermore, in order to visit $q_1$, the balance has to drop once.
  Therefore, between any two (not necessarily successive) visits to the final state $q_1$, one of the following holds:
  (i)~the counter strictly decreases or (ii)~the input word $v$ satisfies $\varphi_1(v)<0$.
  Since $q_1$ is visited infinitely often, we can decompose $w=uv_1v_2\ldots$ such that 
  after reading $v_\ell$, we are in $(q_1,x_\ell)$ and we have $x_1\le x_2\le\cdots$.
  Then ``(i)'' cannot happen for any $v_\ell$ and thus we have $\varphi_1(v_\ell)<0$ for every $\ell$.
  Hence, the claim is proven.
  
  It remains to show $\notregsep{L(\cV)}{D_1}$. Towards a contradiction, suppose
  there is a B\"{u}chi automaton $\cA$ with $n$ states such that $L(\cV)\subseteq
  L(\cA)$ and $L(\cA)\cap D_1=\emptyset$. Note that $\cV$ accepts
  $(a_1^n\bar{a}_1^{n+1})^\omega$: We drive up the counter to $n$ in $q_0$ and
  then read each $a_1^n\bar{a}_1^{n+1}$ in a loop from $q_1$ to $q_1$.  However,
  a run of $\cA$ must cycle on some non-empty infix of $a_1^n$ and thus, for some
  $m>n$, also accept $w=(a_1^{m}\bar a_1^{n+1})^\omega$.  Since $w\in D_1$, that
  is a contradiction.
\end{proof}
The second part of \cref{Theorem:CounterExamples} is proven in \cref{appendix-intersection}.

\subsection{Proof of \cref{one-language-fixed}}
\oneLanguageFixed*
For the proof, we need the concept of rational transductions of infinite words.

\subparagraph{Rational transductions} 
A \emph{finite state B\"{u}chi transducer} is a tuple
$\cT = (Q,\Sigma,\Gamma,E,q_0,Q_f)$ consists of a finite set of \emph{states} $Q$,
an input alphabet $A$, an \emph{initial state} $q_0 \in Q$, a set of \emph{final states} 
$Q_f \subseteq Q$, and a \emph{transition relation} $E \subseteq Q \times (\Sigma \cup \{\varepsilon\})\times(\Gamma\cup\{\varepsilon\}) \times Q$. 
For a \emph{transition} $(q,a,b,q') \in E$, we also write 
$q \xrightarrow{(a,b)} q'$.
The transducer $\cT$ \emph{recognizes} the binary relation
$T(\cT) \subseteq \Sigma^\omega\times\Gamma^\omega$
containing precisely those pairs $(u,v)\in\Sigma^\omega\times\Gamma^\omega$, for which there is a transition sequence
\[ q_0 \xrightarrow{(a_1,b_1)} q_1 \xrightarrow{(a_2,b_2)} \ldots  \]
such that $u=a_1a_2\cdots$, $v=b_1b_2\cdots$, and for infinitely many $i\in\N$, we have $q_i\in Q_f$.
We say that a relation $T\subseteq\Sigma^\omega\times\Gamma^\omega$ is \emph{rational} if there is a finite-state B\"{u}chi transducer $\cT$ with $T=T(\cT)$.
For a language $L\subseteq\Gamma^\omega$ and a relation $T\subseteq\Sigma^\omega\times\Gamma^\omega$, we define
\[ TL=\{u\in\Sigma^\omega \mid \exists v\in L\colon (u,v)\in T\}. \] 
Moreover, for relations $T\subseteq\Sigma^\omega\times\Gamma^\omega$ and $S\subseteq\Theta^\omega\times\Sigma^\omega$, we define
\[ S\circ T =\{ (u,w)\in \Theta^\omega\times\Gamma^\omega \mid \exists v\in\Sigma^\omega\colon (u,v)\in S,~(v,w)\in T\}. \]
Using a simple product construction, we observe that for rational transductions
$S$ and $T$, the relation $S\circ T$ is (effectively) rational as well.
By simply exchanging the two input coordinates, one can also show that if $T\subseteq\Sigma^\omega\times\Gamma^\omega$ is rational, then so is
\[ T^{-1}=\{(u,v)\in\Gamma^\omega\times\Sigma^\omega \mid (v,u)\in T\}. \]
The following is also entirely straightforward.
\begin{lemma}\label{vass-vs-transducers}
A language $L\subseteq\Sigma^\omega$ is a B\"{u}chi VASS language if and only
if there exists a rational transduction $T$ and a number $n\in\N$ such that
$L=TD_n$. Moreover, the translation can be performed in exponential time.
\end{lemma}
Here, the automaton underlying an $n$-dim.\ B\"{u}chi VASS is translated into a
transducer with input in $D_n$ and vice-versa.
More precisely, for $h \in \N$ an operation of $+h$ on the $i$th counter is translated
into the string $(a_i)^h$, whereas $-h$ is translated into $(\bar{a}_i)^h$. 
The $\bzero$-vector is hereby translated into $a_1\bar{a}_1$ instead of $\varepsilon$,
to ensure that every infinite run of the B\"{u}chi VASS actually corresponds to
an infinite word in $D_n$.
The only reason why this construction is not feasible in polynomial time, is because we assume that counter operations of B\"{u}chi VASS are encoded in binary.
In particular, the string $(a_i)^h$ mentioned above takes $h$ steps to write down, whereas the size of the B\"{u}chi VASS is only dependent on $\log{h}$.
However, the construction only takes polynomial time, if counter updates are encoded in unary, or if strings such as $(a_i)^h$ are subjected to some exponential compression.

We also need the following lemma. The proof is exactly the
same as the corresponding proof in~\cite{CzerwinskiZetzsche2020a}. The only
difference is that we have infinite instead of finite words.
\begin{lemma}\label{movetrans}
Let $T\subseteq\Sigma^\omega\times\Gamma^\omega$ be rational and $L\subseteq\Sigma^\omega$ and $K\subseteq\Gamma^\omega$. Then $\regsep{L}{TK}$ if and only if $\regsep{T^{-1}L}{K}$.
\end{lemma}
\begin{proof}
Suppose $L\subseteq R$ and $R\cap TK=\emptyset$ for some regular
$R$. Then clearly $T^{-1}L\subseteq T^{-1}R$ and $T^{-1}R\cap
K=\emptyset$. Therefore, the regular set $T^{-1}R$ witnesses
$\regsep{T^{-1}L}{K}$. Conversely, if $\regsep{T^{-1}L}{K}$, then
$\regsep{K}{T^{-1}L}$ and hence, by the first direction,
$\regsep{(T^{-1})^{-1}K}{L}$. Since $(T^{-1})^{-1}=T$, this reads
$\regsep{TK}{L}$ and thus $\regsep{L}{TK}$.
\end{proof}
We are now ready to prove \cref{one-language-fixed}. 
\begin{proof}[Proof of \cref{one-language-fixed}]
Given B\"{u}chi VASS $\cV_1$ and $\cV_2$, where $\cV_2$ is $n$-dimensional,
\cref{vass-vs-transducers} allows us to compute in exponential time a rational
transduction $T$ such that $L(\cV_2)=TD_n$. We apply \cref{vass-vs-transducers} again to construct a
B\"{u}chi VASS $\cV$ for $T^{-1}L(\cV_1)$. Then we have 
\[ \regsep{L(\cV_1)}{L(\cV_2)} \iff \regsep{L(\cV_1)}{TD_n} \iff \regsep{T^{-1}L(\cV_1)}{D_n}\iff \regsep{L(\cV)}{D_n}, \]
where the second equivalence is due to \cref{movetrans}.
\end{proof}

\section{Proof Details for Pumpability}\label{appendix-pumping}
Let us formally define the constructions of $\bar{\cV}$ and $\pump{\cV}$.

\begin{definition}
  Let $\cV=(Q, q_0, T, F)$ be a $d$-dimensional B\"{u}chi VASS over $\Sigma_n$.
  Then $\bar{\cV}=(Q, q_0, \bar{T}, F)$ is the $(d+n)$-dimensional B\"{u}chi VASS
  over $\Sigma_n$ with transitions constructed as follows:
  $(q, \varepsilon, (\delta,\varphi(w)), q')\in \bar{T}$ if and only if
  $(q, w, \delta, q') \in T$.
  
  Furthermore, $\pump{\cV}=(\pump{Q}, q_{\mathsf{pump},0}, \pump{T}, \pump{F})$
  is the $d$-dimensional B\"{u}chi VASS over $\Sigma_n$ constructed as follows:
  \begin{itemize}
    \item $\pump{Q} = Q \times (\mathbb{N}\cup\{\omega\})^{d+n}$, i.e. the states of $\KM{\bar{\cV}}$,
    \item $q_{\mathsf{pump},0} = (q_0,\bzero)$,
    \item $\pump{F} = F \times (\mathbb{N}\cup\{\omega\})^{d+n}$, and
    \item $(\pump{q}, w, \delta, \pump{q'}) \in \pump{T}$ if and only if there is a transition
    $(\pump{q}, t, \pump{q'})$ in $\KM{\bar{\cV}}$ labelled by
    $t = (q, \varepsilon, (\delta,\varphi(w)), q')\in \bar{T}$.
  \end{itemize}
\end{definition}

We split the first two parts of \cref{make-pumpable} into
\cref{pump-limsep,pump-inclusion}, which we prove separately.
The third part later follows from \cref{prefix-independent-limsep},
which, in turn, follows from \cref{thm-basic-separators}.

\begin{lemma}\label{pump-limsep}
	$L(\pump{\cV})$ is pumpable.
\end{lemma}
\begin{proof}
	Consider some $w\in L(\pump{\cV})$ and some $k\in\N$.  Let $\rho$ be an
	accepting run of $\cV$ over $w$.  By construction of $\pump{L}$, there
	exists a corresponding run $\bar{\rho}$ in $\KM{\bar{\cV}}$. Let
	$\Omega\subseteq[d+1,d+n]$ be the set of coordinates where the states
	of $\KM{\bar{\cV}}$ carry $\omega$ eventually during $\bar{\rho}$.  Then at
	some point, $\bar{\rho}$ visits an extended configuration $(q,\bar{\bmm})\in
	Q\times\N_\omega^{d+n}$ where all coordinates from $\Omega$ in
	$\bar{\bmm}$ are $\omega$. Decompose
	$\bar{\rho}=\bar{\rho}_0\bar{\rho}_1$ so that $\bar{\rho}_0$ reaches
	$(q,\bar{\bmm})$. Let $\rho=\rho_0\rho_1$  and $w=w_0w_1$ be the
	corresponding decompositions of $\rho$ and $w$.  Then $\rho_0$ reaches
	a configuration $(q,\bmm)\in Q\times\N^d$ in $V$. 

  Let $\ell = \max_{i \in [1,n]}\{0,-\varphi_i(w_0)\} + k$.
	By the construction of Karp-Miller graphs, there exists a run
	$\bar{\rho}'_0$ in $\bar{\cV}$ that reaches a configuration
	$(q,\bar{\bmm}')\in\N^{d+n}$ such that $\bar{\bmm}'(i)\ge\bmm(i)$ for
	$i\in[1,d+n]\setminus\Omega$, and $\bar{\bmm}'(i)>\bmm(i)+\ell$ for $i\in \Omega$.
	Then
	$\bar{\rho}'_0$ corresponds to a run $\rho'_0$ in $\cV$. It reaches a
	configuration $(q,\bmm')$ with $\bmm'\ge\bmm$ and thus $\rho'_0\rho_1$ is
	a run of $\cV$.  It reads a word $w'_0w_1\in \pump{L}$, where $w'_0$ is the prefix
	read by $\rho'_0$. Since $w'_0$ was also read by $\bar{\rho}'_0$ in $\bar{\cV}$, it
  is a prefix of some word in $D_n$, as mandated by the additional counters of $\bar{\cV}$.

	We claim that now $\varphi(w'_0)\ge \varphi(w_0)$ and for every
	$i\in[1,n]$ where $w$ ever becomes negative, we have
	$\varphi_i(w'_0)\ge\max\{\varphi_i(w_0),0\}+k$.  The first condition follows
	from the fact that
	$\varphi_i(w'_0)=\bmm'(d+i)\ge\bmm(d+i)=\varphi_i(w_0)$. For the second
	condition, note that if $w$ ever becomes negative in coordinate $i$,
	then $\bar{\rho}$ must necessarily visit a configuration where in
	coordinate $i$, there is an $\omega$. In particular, we have
	$d+i\in\Omega$ and thus $\varphi_i(w'_0)=\bmm'(d+i)\ge
	\bmm(d+i)+\ell=\varphi_i(w_0)+\max_{i \in [1,n]}\{0,-\varphi_i(w_0)\} + k
  \ge\max\{\varphi_i(w_0),0\}+k$.
\end{proof}

\begin{lemma}\label{pump-inclusion}
	There exists a $k\in\N$ such that $L(\pump{\cV})\subseteq L(\cV)\subseteq  L(\pump{\cV}) \cup P_k$.
\end{lemma}
\begin{proof}
	The inclusion $L(\pump{\cV})\subseteq L(\cV)$ is obvious from the
	construction.  For the second inclusion, define $k \in \N$ to be the largest
	number occurring in the states of $\KM{\bar{\cV}}$.  We claim that then
	$L(\cV)\subseteq P_k\cup L(\pump{\cV})$.  Let $w\in L$ be accepted by a
	run $\rho$ in $\cV$ and suppose $w\notin P_{i,k}$ for some $i\in[1,n]$.
	If $u$ is a prefix of $w$, then we say that $i\in[1,n]$ is
	\emph{crossing at $u$} if $\exteff_i(u)<0$ and $\exteff_i(v)\ge 0$
	for every prefix $v$ of $u$.  Observe that whenever $i$ is crossing at
	$u$, then $\exteff_i(v)>k$ for some prefix $v$ of $u$: Otherwise, $w$
	would belong to $P_{i,k}$.  This implies that $\rho$ has a
	corresponding run in $\KM{\bar{\cV}}$: Whenever a counter in
	$[d+1,d+n]$ drops below zero, it must have been higher than $k$ before
	and thus been set to $\omega$. Therefore, $w$ is also accepted by
	$\KM{\bar{\cV}}$ and thus $w\in L(\pump{\cV})$.
\end{proof}

\section{Proof Details for Basic Separators}\label{appendix-basic-separators}

\begin{proof}[Proof of \cref{lem-basic-post-separators}]
  First of all, if $L_\pi(\cA)$ is empty, then Condition (\labelcref{post-separator-exists}) trivially holds. Thus, in the following we assume that $L_\pi(\cA) \neq \emptyset$ and in particular that the final state $q_\pi$ associated with the profile $\pi$ is reachable from $\cA$'s initial state.
  
  We want to set up a system of linear inequalities that has a solution $\bx$ if and only if there is a $k$ such that $L_\pi(\cA) \subseteq S_{\bx,k}$.
  Therefore let us talk about some requirements that are necessary for the above inclusion to hold.
  These requirements will be on cycles of transitions in $\pi$, and we make sure that they can be expressed as linear inequalities.
  
  The cycle $\sigma_\pi$ that contains exactly the transitions in $\pi$ has to be over a word $v_\pi \in \Sigma_n^*$ with $\langle\bx, \varphi(v_\pi)\rangle \leq -1$.
  Otherwise, we can just repeat $\sigma_\pi$ infinitely often and prepend any prefix leading to $q_\pi$ from $\cA$'s initial state, yielding a word that violates requirement~b.) of $S_{\bx,k}$.
  Any primitive cycle $\sigma$ in $\cA$ of transitions in $\pi$ has to be over a word $v$ with $\langle\bx, \varphi(v)\rangle \leq 0$.
  Otherwise we repeat $\sigma_\pi$ infinitely often from some arbitrary prefix reaching $q_\pi$ like before, and then perform $k+1$ insertions of the cycle $\sigma$ into each copy of $\sigma_\pi$.
  This yields a word that violates requirement~a.) of $S_{\bx,k}$, and we can do this for any $k$.
  
  We can now use these requirements on cycles to construct a linear system of inequalities $\bA_\pi\bx \leq \bb$ for $\bx$.
  For the cycle $\sigma_\pi$ corresponding to word $v_\pi \in \Sigma_n^*$, we add the inequality
  \[
  x_1\varphi_1(v_\pi) + x_2\varphi_2(v_\pi) + \ldots + x_n\varphi_n(v_\pi) \leq -1,
  \]
  and for each primitive cycle $\sigma$ of transitions in $\pi$ over a word $v \in \Sigma_n^*$, we add the inequality
  \[
  x_1\varphi_1(v) + x_2\varphi_2(v) + \ldots + x_n\varphi_n(v) \leq 0.
  \]

  Let us argue that the precise choice of the justifying cycle $\sigma_\pi$ does not affect the satisfiability of the system $\bA_\pi\bx \leq \bb$.
  To this end we argue that $\bx \in \N^n$ is a valid solution to the system if and only if (1) all primitive cycles have an $\bx$-weighted balance at most zero, and (2) at least one primitive cycle has a strictly negative $\bx$-weighted balance. 
  Constraint (1) is clearly equivalent to the inequalities added for each primitive cycle.

  For constraint (2), assume that the inequality $\langle\bx, \varphi(v_\pi)\rangle \leq -1$ holds.
  Now observe that any valid choice of $\sigma_\pi$ is a cycle and therefore can be constructed by inserting primitive cycles into each other a finite number of times.
  If all primitive cycles had non-negative $\bx$-weighted balance, then the $\bx$-weighted balance for $\sigma_\pi$ could not be negative.

  For the other implication direction, assume that constraint (2) holds, and let the primitive cycle with negative $\bx$-weighted balance be $\sigma'$.
  Since any valid choice of $\sigma_\pi$ contains each transition in $\pi$, its $|\pi|$-fold repetition $\sigma_\pi^|\pi$ contains each primitive cycle as a (possibly non-contiguous) subsequence.
  Now, if we delete $\sigma'$ from $\sigma_\pi^|\pi$, the remaining (possibly not connected) transition sequences still combine to form a collection of cycles, since $\sigma'$ is a cycle.
  Thus, the summed-up $\bx$-weighted balance of this collection is the sum of $\bx$-weighted balances of primitive cycles, and can therefore be at most zero by Condition (1).
  Then adding $\sigma'$ back in gives us that $\langle\bx, \varphi(v_\pi^|\pi|)\rangle$ is negative, and therefore $\langle\bx, \varphi(v_\pi)\rangle$ is as well.
  Since the letter balance can only have integer-values, and weighting by $\bx \in \N$ does not change this, it follows that $\langle\bx, \varphi(v_\pi)\rangle \leq -1$.
  
  The characterization of solutions $\bx$ via constraints (1) and (2) is clearly independent of $\sigma_\pi$, meaning the precise choice of the latter does not affect satisfiability of the system $\bA_\pi\bx \leq \bb$.
  Furthermore, the restriction of $\bx \in \N^d$ is not a meaningful one, as we can always compute a solution in $\Q^d$ from one in $\N^d$, as we explain below.
  
  Applying Farkas' Lemma (\cref{farkas-lemma}) to this system of equations, we either obtain a vector $\bx \in \Q_{\geq0}^n$ as a suitable solution, or we obtain a vector $\by \in \Q_{\geq0}^{m}$ with $\by^\top \bA_\pi \geq \bzero^\top$ and $\by^\top \bb < 0$, where $m$ is the number of rows of $\bA_\pi$.
  
  In the first case, we can multiply the entries of $\bx$ by their denominators' least common multiple, say $\ell$, to yield a suitable vector $\ell \cdot \bx = \bx' \in \N^n$.
  Furthermore we set $k = |Q_\pi| \cdot h$, where $|Q_\pi| \supseteq \{q_\pi\}$ is the set of all states of $\cA$ adjacent to transitions in $\pi$, and $h$ is the length of the longest word appearing as a transition label of $\cA$.
  With this we can show that $L_\pi(\cA) \subseteq S_{\bx',k}$:
  Each word $w \in L_\pi(\cA)$ decomposes into $uv$ with $v = v_0v_1v_2\cdots$ such that $u$ leads to $q_\pi$ from $\cA$'s initial state and each $v_j$ corresponds to some cycle $\sigma_j$ on $q_\pi$, that contains each transition of $\pi$ at least once.
  Then we have $\langle\bx', \varphi(v_j)\rangle = \ell \cdot \langle\bx, \varphi(v_j)\rangle < \ell \cdot 0 = 0$ as required by $S_{x,k}$:
  each cycle $\sigma_j$ can be obtained by starting with $\sigma_\pi$, which contributes at most $-1$ to this value, and inserting finitely many primitive cycles, which all add at most $0$.
  Moreover, we need to show $\langle\bx', \varphi(f)\rangle \leq k$ for every infix $f$ of $v$.
  Towards a contradiction assume there is at least one infix $f$ of $v$, for which this does not hold.
  Since $f$ fulfils $\langle\bx', \varphi(f)\rangle > |Q_\pi| \cdot h$, and $h$ is the maximum length of a transition label, the transition sequence corresponding to $f$ has to be longer than $|Q_\pi|$.
  Thus this sequence repeats a state and therefore has to contain a primitive cycle.
  However, all such primitive cycles add at most $0$ to the value $\langle\bx', \varphi(f)\rangle$, meaning one could delete the word corresponding to this cycle from $f$ and still fulfil the aforementioned requirement.
  One can repeatedly remove primitive cycles until one obtains a word $f'$ of length $|f'| \leq |Q_\pi| \cdot h$ with $\langle\bx', \varphi(f')\rangle > |Q_\pi| \cdot h$.
  This is a contradiction, therefore infixes such as $f$ cannot exist.
  
  In the other case we also multiply $\by$ with the least common multiple of its entries, say $\ell$, to yield $\ell \cdot \by = \by' \in \N^{m}$.
  Furthermore, each row of the matrix $\bA_\pi$ essentially contains the $\varphi$-values of its corresponding cycle.
  The requirement ${\by'}^\top \bA_\pi = \ell \by^\top \bA_\pi \geq \ell \cdot \bzero^\top = \bzero^\top$ can then be seen as a selection of cycles, whose combined $\varphi$-values are all $0$ or above.
  Moreover, the requirement ${\by'}^\top \bb = \ell \by^\top \bb < \ell \cdot 0 = 0$ ensures that $\sigma_\pi$ is selected at least once, because all other entries of $\bb$ are $0$, meaning we would have $\by^\top \bb = 0$ if $\sigma_\pi$ was not selected.
  This means we can combine all the selected cycles into one large cycle $\sigma'$ via matching states, which is possible because $\sigma_\pi$ visits all states in $Q_\pi$.
  Since the combined $\varphi$-values of all the cycles selected by $\by$ are $0$ or above, we have that $\sigma'$ corresponds to a word $w'$ with $\varphi(w') \geq 0$.
  Finally, $\sigma'$ also contains all transitions of $\pi$ as required, because it contains the cycle $\sigma_\pi$.
\end{proof}

\subparagraph{Regarding \cref{thm-basic-separators}}
We mentioned in \ref{sec:basic-separators} that a single value of $k$ is sufficient for a finite union of basic separators $P_{i,k}$ and $S_{\bx,k}$.
This is because we have $P_{i,k} \subseteq P_{i,k+1}$ and $S_{\bx,k} \subseteq S_{\bx,k+1}$ for each $i \in [1,n], \bx \in \N^n, k \in \N$.
Therefore it suffices to show the following:

{\itshape
  Let $\cA$ be a B\"{u}chi automaton with $L(\cA) = R \subseteq \Sigma_n^\omega$ and $R \cap D_n = \varnothing$. Then there is a finite set $X \subseteq \N^n$ and a number $k \in \N$ such that $R \subseteq \bigcup_{i \in [1,n]} P_{i,k} \cup \bigcup_{\bx \in X} S_{\bx,k}$.
}

Here $R$ is a \emph{separator candidate} in the sense of the original phrasing of the theorem, because it is $\omega$-regular and disjoint from $D_n$.

\begin{proof}[Proof of \cref{thm-basic-separators}]
  We begin by invoking \cref{make-pumpable} on $\cA$ to obtain a B\"{u}chi automaton $\pump{\cA}$, whose language is pumpable, and a number $\ell$ such that $L(\pump{\cA}) \subseteq L(\cA) \subseteq L(\pump{\cA}) \cup P_{\ell}$.
  Using the theorem this way is feasible, because B\"{u}chi automata can be seen as $0$-dimensional B\"{u}chi VASS.
  Since $L(\cA) \cap D_n = \emptyset$ and $L(\pump{\cA}) \subseteq L(\cA)$ we have $L(\pump{\cA}) \cap D_n = \emptyset$.
  It now suffices to show that the basic separators theorem holds for $L(\pump{\cA})$:
  If there are $X,k$ such that $L(\pump{\cA}) \subseteq \bigcup_{i \in [1,n]} P_{i,k} \cup \bigcup_{\bx \in X} S_{\bx,k}$ then $L(\cA) \subseteq L(\pump{\cA}) \cup P_{\ell} \subseteq \bigcup_{i \in [1,n]} P_{i,o} \cup \bigcup_{\bx \in X} S_{\bx,o}$, where $o = \max(k,\ell)$.
  
  Now consider the decomposition $L(\pump{\cA}) = \bigcup_{\pi \in \Pi(\pump{\cA})} L_\pi(\pump{\cA})$. If we can show that each language $L_\pi(\pump{\cA})$ is contained in a finite union of basic separators, then we are done. In the following let us fix a profile $\pi$ of $\pump{\cA}$.
  
  We now invoke \cref{lem-basic-post-separators} on $\pump{\cA}$ and $\pi$. If Condition (\labelcref{post-separator-exists}) holds, then this already yields $\bx,k$ such that $L_\pi(\pump{\cA}) \subseteq S_{\bx,k}$, and we need not concern ourselves with the languages $P_{i,k}$.
  
  In the other case, Condition (\labelcref{post-separator-nonex-cycle}) yields a cycle $c'$ in $\pump{\cA}$ that contains all transitions in $\pi$ and is over a word $w'$ with $\varphi(w') \geq 0$. Since Condition (\labelcref{post-separator-exists}) did not hold, we know that $L_\pi(\pump{\cA})$ is not empty, which means that all states adjacent to transitions of $\pi$ are reachable from $\cA$'s initial state, including the final state $q_\pi$ associated with $\pi$. Let $u'$ be a word that reaches $q_\pi$ from $\pump{\cA}$'s initial state. Then $\tilde{w} = u'(w')^\omega \in L(\pump{\cA})$.
  
  Let $m$ be the lowest value of $\varphi_i$ for any index $i$ and prefix of $\tilde{w}$, formally $m = \min_{i \in [1,n], v \in \prefix(\tilde{w})}\varphi_i(v)$.
  Since $\varphi(w') \geq 0$ we know that $m \in \Z$ is well-defined.
  Moreover, since $L(\pump{\cA})$ is pumpable, there is a decomposition $\tilde{w} = u_0w_1$ and a word $v_0 \in \Sigma_n^*$ such that $v_0w_1 \in L(\pump{\cA})$, $\varphi(v_0) \geq \varphi(u_0)$, and $\varphi_i(v_0) \geq \varphi_i(u_0) + |m|$ for all indices $i$ where there is a $v \in \prefix(\tilde{w})$ with $\varphi_i \geq 0$.
  Then swapping $u_0$ for $v_0$ in $\tilde{w}$ can only increase the $\varphi_i$-values of its prefixes, and in fact all such values that fell below $0$ are now raised above $0$ by choice of $|m|$.
  This means that $v_0w_1 \in L(\pump{\cA}) \cap D_n$, which is a contradiction.
\end{proof}

\section{Proof Details for Decidability}\label{appendix-decidability}
\prefixIndependentLimSup*
\begin{proof}
	The ``if'' direction is trivial. Conversely, let  $\regsep{L}{D_n}$.
	By \cref{thm-basic-separators}, we have $L\subseteq\bigcup_{i\in[1,n]}
	P_{i,k}\cup\bigcup_{\bx\in X} S_{\bx,k}$ for some finite
	$X\subseteq\N^n$ and $k\in\N$. We claim that $L\subseteq
	\bigcup_{\bx\in X}S_{\bx,k}$, which yields $\limsep{L}{D_n}$. Indeed,
	given $u\in L$, pumpability yields a word $u'\in L$ such that $u'\sim
	u$ and $u'\notin P_{i,k}$ for any $i\in[1,n]$. Since $u'\in L\subseteq
	P_k \cup \bigcup_{\bx\in X} S_{\bx,k}$, we conclude $u'\in S_{\bx,k}$.
	Finally, observe that membership in $S_{\bx,k}$ is not affected by
	changing a finite prefix of a word. Therefore, we also have $u\in
	S_{\bx,k}$.
\end{proof}

\begin{restatable}{lemma}{xYieldsSeparator}\label{x-yields-separator}
	Let $\pi\in\prof{\cV}$. If $\bA_\pi\bx\le\bb$ for $\bx\in\N^n$, then
	$L_\pi(\cV)\subseteq S_{\bx,k}$ for some $k\in\N$.
\end{restatable}
\begin{proof}
	We regard $\KM{\cV}$ as a B\"{u}chi automaton.  Then, $\pi$ is in
	particular a profile for $\KM{\cV}$. Moreover, the cycle witnessing
	that $\pi$ is a profile is also an admissible cycle for $\pi$ in
	$\KM{\cV}$ as a B\"{u}chi automaton. Thus, \cref{solution-implies-separator} implies 
	$L_\pi(\cV)\subseteq L_\pi(\KM{\cV})\subseteq S_{\bx,k}$ for some $k\in\N$. 
\end{proof}

\section{Proof Details for One-dimensional B\"{u}chi VASS}\label{appendix-one-dim}
\subsection{\cref{one-dim-complexity}: $\PSPACE$-hardness}\label{appendix-one-dim-hardness}
We begin with the straightforward reduction from intersection emptiness of
finite-word languages of $1$-dim.\ VASS. Suppose $L_1,L_2\subseteq\Sigma^*$ are
finite-word languages of $1$-dim.\ VASS with succinct counter updates, and
acceptance by final state. Checking whether the intersection $L_1\cap L_2$ is
empty is $\PSPACE$-complete~\cite{DBLP:journals/iandc/FearnleyJ15}.  We
construct B\"{u}chi VASS for $L_1\#^\omega$ and~$L_2\#^\omega$, where $\#$ is a
fresh letter.  Since $L_1$ and $L_2$ are coverability languages of finitely-branching WSTS,
we know from \cite[Theorem 7]{WSTSRegSep2018} that $L_1\cap L_2=\emptyset$ if
and only if $\regsep{L_1}{L_2}$.  Furthermore, with a fresh letter $\#$, it is
easy to observe that $\regsep{L_1}{L_2}$ if and only if
$\regsep{L_1\#^\omega}{L_2\#^\omega}$. 

\subparagraph{Hardness for disjoint languages}
In the $\PSPACE$-hardness proof above, one can notice that the languages $L_1\#^\omega$
and $L_2\#^\omega$ are regularly separable if and only if they are disjoint.
 In order to further highlight the disparity between
the finite-word case of WSTS languages (where disjointness and separability coincide~\cite{WSTSRegSep2018}) and the infinite-word case, we want to present a proof that $\PSPACE$-hardness already
holds if the input languages are promised to be disjoint: Note that with this promise,
separability in the finite-word case becomes trivial.

Here, we reduce directly from configuration reachability in bounded one-counter automata, which was shown to be PSPACE-hard in \cite{FearnleyJurdzinski13}.

  A bounded one-counter automaton $B=(\cV_B,b)$ consists of a $1$-dim. VASS $\cV_B$ equipped with a bound $b \in \N$ on its counter values.
  This means transitions of $B$ are enabled if and only if they meet the firing restrictions of a VASS and also lead to a configuration $(q,m)$ with $m \leq b$.
  Here, counter values and the bound $b$ are encoded in binary.
  In particular, the size of $B$ is that of the underlying VASS plus $\log{b}$, and the size of a configuration $(q,m)$ is $\log{m}$.
  
  Now we want to construct two $1$-dim.\ B\"{u}chi VASS $\cV_1$ and $\cV_2$, whose languages are always disjoint, but are $\omega$-regular separable if and only if $(q,m)$ is not reachable from $(q_0,0)$ in $B = (\cV_B,b)$.
  Let $T$ be the set of transitions of $B$.
  We use $\Sigma = T \cup \{\#\} \cup \Sigma_1$ as the alphabet for $\cV_1$ and $\cV_2$.
  Let $\cV_{D1}$ be the $1$-dim.\ B\"{u}chi VASS accepting $D_1$, i.e.\ $\cV_{D1}$ consists of a single state, both initial and final, with two loops $\be_1 | a_1$ and $-\be_1 | \bar{a}_1$.
  Furthermore let $\cV_S$ be the $1$-dim.\ B\"{u}chi VASS from \cref{Figure:examples}(left) accepting the language $S$ with $S\cap D_1=\emptyset$ but $\notregsep{S}{D_1}$, which we talked about in \cref{Section:Outline} (see the proof of the first statement in \cref{Theorem:CounterExamples}).
  
  We start constructing $\cV_1$ by using a copy of $\cV_B$ with all states being non-final and every transition $t \in T$ labelled with $t$ itself.
  Then we add a copy of $\cV_{D1}$ with its only state still being final.
  To connect the two copies, we add the transition $-m | \#$ from state $q$ of $\cV_B$ to the initial state of $\cV_{D1}$.
  
  For $\cV_2$ we also start with a copy of $\cV_B$ with all non-final states and transitions labeled with themselves, but we also invert every transition effect, changing it from $z \in \mathbb{Z}$ to $-z$.
  Then we add a new initial state~$q_0'$ with the same outgoing transitions as the initial state of $\cV_B$, except we change their original effects $z$ to $b - z$.
  These new transitions of $q_0'$ are labelled with their original copies from~$T$.
  Additionally, we add a copy of $\cV_S$ with $q_2$ still being a final state.
  The two copies are then connected with a transition $m - b | \#$ from $q$ to the initial state of $\cV_S$.
  If $(q,m) = (q_0,0)$, we also add a transition $0 | \#$ from $q_0'$ to the initial state of $\cV_S$.
  
  Now let $R_1$ be the set of all transition sequences over $T$ that cover $(q,m)$ in $\cV_B$ and do not necessarily respect the bound $b$. 
  Formally, $\rho \in R_1$ if $\rho$ leads from $(q_0,0)$ to $(q,m')$ in $\cV_B$ for some $m' \in \N$ with $m' \geq m$.
  Moreover let $R_2$ be the set of all transition sequences over $T$ that reach $q$ with a counter value below $m$, when respecting the upper bound $b$, but not necessarily the lower bound $0$ of VASS counters.
  Formally, $\rho \in R_2$ if $\rho$ leads from $(q_0,0)$ to $(q,z')$ in $B'$ for some $z' \in \mathbb{Z}$ with $z' \leq m$, where $B' = (\cV',b)$ and $\cV'$ is just $\cV$ interpreted as a $\Z$-VASS.
  We now want to argue that there are languages $L_1, L_2$ such that $L(\cV_1) = R_1\#L_1$ and $L(\cV_2) = R_2\#L_2$.
  
  $L(\cV_1) = R_1\#L_1$ is easy to see, since $\cV_1$ simulates $\cV_B$ faithfully, and can only read $\#$ if a configuration $(q,m)$ or greater is reached.
  For $L(\cV_2) = R_2\#L_2$ observe that before reading $\#$, $\cV_2$ essentially simulates $\cV_B$ with inverted counter values, starting with $b$ instead of $0$.
  Since $\cV$ can go above $b$, this essentially simulates going below $0$ in $B$.
  The $\#$ can also only be read if in $q$ the counter is valued at least $m - b$, which corresponds to at most $m$ before inversion.
  Let us now show that $L(\cV_1) \cap L(\cV_2) = \emptyset$, and furthermore $\regsep{L(\cV_1)}{L(\cV_2)}$ if and only if $(q_0,0) \rightarrow (q,m)$ in $B$.
  
  The configuration $(q,m)$ not being reachable in $B$ is equivalent to $R_1$ and $R_2$ being disjoint.
  In this case $L(\cV_1)$ and $L(\cV_2)$ also have to be disjoint, since the prefixes before the '$\#$' of their words cannot coincide.
  They are however $\omega$-regular separable:
  With $Q$ being the states of $B$, an exponential size B\"{u}chi automaton $A$ with states $Q \times \{0,\ldots,b\}$ can simulate $B$.
  To make $A$ accept all words with prefixes in $R_1$, we add a final state with loops on all input letters, that is reachable by every transition that would make the counter value go above $b$.
  Now $L(\cV_1) \subseteq L(A)$ is clear. %
  Since transition sequences that do not respect the bound $b$ cannot be prefixes of elements of $R_2$, $L(A) \cap L(\cV_2) = \varnothing$ immediately follows.
  Thus, $L(A)$ is an $\omega$ regular separator for $L(\cV_1)$ and $L(\cV_2)$, which also means that they are disjoint.
  
  If $(q,m)$ is reachable in $B$, we have a finite transition sequence $\rho \in R_1 \cap R_2$.
  Reading $\rho$ then leads to $(q,m)$ in $\cV_1$, respectively to $(q,b-m)$ in $\cV_2$.
  Therefore if $\#$ is read right after, the counter value of either B\"{u}chi VASS would be $0$.
  This implies that $\rho{}\#D_1 \subseteq L(\cV_1)$ and $\rho{}\#S \subseteq L(\cV_1)$, as the second component of either VASS would be simulated faithfully after this prefix.
  A regular separator $A$ for $L(N_1)$ and $L(N_2)$ would therefore have to accept all words in $\rho{}\#D_1$ but no words in $\rho{}\#S$.
  By adding a new initial state $q_{init}$ to $A$ and adding all outgoing transitions of states reachable via $\rho{}\#$ in the original $A$ to $q_{init}$, we obtain an $\omega$-regular separator for $D_1$ and $S$.
  This is a contradiction, since we established earlier that these languages are not $\omega$-regular separable as shown in the proof of the first half of \cref{Theorem:CounterExamples}.
  
  It remains to show that $L(\cV_1) \cap L(\cV_2) = \emptyset$ in the case where $(q,m)$ is reachable in $B$.
  For transition sequences $\rho$ over $T$, we know that $\rho \in R_1 \cap R_2$ if and only if $\rho$ reaches $(q,m)$ in $B$.
  Therefore two words $w_1 \in L(\cV_1)$ and $w_2 \in L(\cV_2)$ can only agree on a prefix $\rho{}\#$, if $\rho$ has this property.
  However, in this case $w_1 = \rho{}\#w_1'$ for some $w_1' \in D_1$ and $w_2 = \rho{}\#w_2'$ for some $w_2' \in S$.
  This yields $w_1 \neq w_2$ since $D_1$ and $S$ are disjoint.

\subsection{Proof of \cref{flower-in-Vbar-sufficient}}
\flowerInVBarSufficient*
\begin{proof}
We first invoke \cref{make-pumpable} to obtain $\pump{\cV}$ with $\notregsep{L(\pump{\cV})}{D_n}$ if and only if $\notregsep{L(\cV)}{D_n}$.
Recall that $\pump{\cV}$ was constructed as the product of $\cV$ and $\KM{\bar{\cV}}$, which means that every cycle of $\pump{\cV}$ is also a cycle of $\KM{\bar{\cV}}$.

For the if direction, we get that $\KM{\pump{\cV}}$ contains an inseparability flower by \cref{dec}.
Its three cycles then correspond to three cycles of $\pump{\cV}$, which then also appear in $\KM{\bar{\cV}}$, where they still fulfill the requirements of an inseparability flower.

For the only if direction, observe that $\KM{\pump{\cV}}$ is essentially the product construction of $\KM{\cV}$ and $\KM{\bar{\cV}}$.
Furthermore, any transition sequence permitted by $\bar{\cV}$ is also permitted by $\cV$, as the former only added restrictions in the form of more counters, but did not remove any.
Thus, each cycle of $\KM{\bar{\cV}}$ (including the ones that make up its inseparability flower)
also appears as a cycle in $\KM{\pump{\cV}}$.
This implies that $\KM{\pump{\cV}}$ also has an inseparability flower, and by \cref{dec} it follows that $\notregsep{L(\pump{\cV})}{D_n}$.
\end{proof}

\subsection{Proof of \cref{Lemma:CompactVBar}}
\compactVBar*
\begin{proof}
  Recall that $\bar{\cV}$ is constructed from a B\"{u}chi VASS $\cV$ over alphabet $\Sigma_n$ by adding $n$ additional counters and for each transition $t$ replacing its label $w \in \Sigma_n^*$ with $\varepsilon$ and instead adding to $t$ an effect of $\varphi(w)$ on the $n$ additional counters.
  A precise definition can be found in \cref{appendix-pumping}.
  To now proof \cref{Lemma:CompactVBar} we have to argue that we can modify \cref{one-language-fixed} to directly construct $\bar{\cV}$ instead of $\cV$,
  and that the modified version is feasible in polynomial time.
  
  If we analyze the proof of \cref{one-language-fixed} in \cref{appendix-problem} then we obtain exponential time complexity for this construction.
  The bottleneck here is \cref{vass-vs-transducers}.
  However, we already mentioned in the proof of \cref{vass-vs-transducers}, that its associated complexity shrinks from exponential to polynomial time, if we can somehow compress the exponentially long transition labels that we end up with.
  An adequate compression for this is replacing an exponentially long string $w \in \Sigma_n^*$ by its effect on the letter balance $\delta(w)$, which is exactly what we do when going from $\cV$ to $\bar{\cV}$.
  Since we encode $\delta(w)$ in binary, this is an exponential compression, and therefore the time complexity of constructing $\bar{\cV}$ directly is only polynomial, as required.
  
  Note that for two $1$-dimensional B\"{u}chi VASS as input, we have $n = 1$.
  But our proof shows that constructing $\bar{\cV}$ in polynomial time would still be feasible for B\"{u}chi VASS of arbitrary dimension $n$.
\end{proof}

\subsection{Proof of \cref{constrained-runs-in-pspace}}
\constrainedRunsInPSPACE*
\begin{proof}
  We show that if there is a constrained run,
  then there is one where all counters have at most exponential values along the way.
  For this, we rely on a result from~\cite{DBLP:journals/jacm/BlondinEFGHLMT21}
  about linear path schemes. 
  
  A \emph{linear path scheme} (LPS) for a $2$-dimensional VASS $\cV$ is a
  regular expression of the form $S = \sigma_0\lambda_1\sigma_1\cdots\lambda_n\sigma_n$.
  Its alphabet is the set $T$ of transition of $\cV$, and each infix $\lambda_i$
  corresponds to a cycle of transitions in $\cV$.
  
  Each LPS $S$ induces a reachability relation $\to_S$ over configurations of
  $\cV$, where $(q,x,y)\to_S(q',x',y')$ if and only if there are numbers
  $x_1,\ldots,x_n\in\N$ such that $\sigma_0\lambda_1^{x_1}\sigma_1\cdots \lambda_m^{x_n}\sigma_n$
  is a run of $\cV$ from $(q,x,y)$ to $(q',x',y')$.
  In \cite[Theorem 3.1]{DBLP:journals/jacm/BlondinEFGHLMT21}, it is shown that for
  any two states $q,q'$ in a $2$-VASS $\cV$, there exists a set $\mathcal{S}$
  of LPSs, each of which is of polynomial length, such that for $x,y,x',y' \in \N$,
  $(q',x',y')$ is reachable from $(q,x,y)$ if and only if $(q,x,y)\to_S(q',x',y')$
  for some $S$ from $\mathcal{S}$.
  
  In \cite{DBLP:journals/jacm/BlondinEFGHLMT21}, this yields a $\PSPACE$ algorithm
  for configuration reachability in $2$-dimensional VASS:
  If there is run reaching a certain configuration, then there
  is one of the form $\sigma_0\lambda_1^{x_1}\sigma_1\cdots\lambda_n^{x_n}\sigma_n$
  for some LPS $\sigma_0\lambda_1\sigma_1\cdots\lambda_n\sigma_n$ of polynomial length.
  Now the fact that $\sigma_0\lambda_1^{x_1}\sigma_1\cdots\lambda_n^{x_n}\sigma_n$ is
  a run between two given configurations can be expressed using a set of linear
  inequalities over $x_1,\ldots,x_n$. Since each solvable polynomial-sized set of
  linear inequalities has a solution with at most exponential entries, this yields
  a run where all counters are at most exponential.
  
  We only need to extend this argument from \cite{DBLP:journals/jacm/BlondinEFGHLMT21}
  slightly: First, we want to guess a system of linear inequalities, whose solutions
  would satisfy the Presburger formula $\psi$.
  To this end, we view $\psi$ as a propositional formula by treating each atomic formula
  as a proposition. For Presburger arithmetic, an atomic formula is either an equality
  $t_1 = t_2$ or an inequality $t_1 < t_2$, where $t_1, t_2$ are additive terms over
  variables and/or the constants $0, 1$.
  With this propositional view of $\psi$, we can guess an assignment to its propositions,
  and verify that its a satisfying assignment, feasible in polynomial space.
  If this assignment sets an atomic formula of the form $t_1 = t_2$ to false, this
  means that $t_1 < t_2$ or $t_2 < t_1$ has to hold. Similarly if $t_1 < t_2$ is set
  to false then $t_1 = t_2$ or $t_2 < t_1$ has to hold. In both cases, we simply guess
  one of the two atomic formulas that have to hold instead. With these guesses together
  with the unchanged atomic formulas that were set to true, we obtain a system of
  equalities and inequalities, whose solutions would satisfy $\psi$.
  Formally, this system is comprised of matrices $\bA \in\Z^{\ell\times m}, \bC \in\Z^{k\times m}$
  and vectors $\bb \in \Z^\ell, \bd \in \Z^k$ with entries encoded in binary, such that
  $x \in \N^m$ is a solution if and only if $\bA\bx < \bb$ and $\bC\bx = \bd$.
  In fact unary encodings would suffice for our definition of Presburger, since an
  entry of e.g.\ $3$ would have come from a term of the form $y + y + y$ for a variable
  $y$, meaning all entries are polynomial in the size of $\psi$. However, we do not require
  unary encodings and can also work with binary ones.
  
  Now we only need to check that there is a constrained run
  $(q_0,0,0)\autsteps(q_1,x_1,y_1)\autsteps\cdots\autsteps(q_m,x_m,y_m)$, whose
  counter values indeed satisfy these equalities and inequalities.
  This is the case if $\bA\bz < \bb$ and $\bC\bz = \bd$, where $\bz=(x_1,y_1,\ldots,x_m,y_m)$.
  If such a constrained run  exists, then for each $i\in[1,m]$, there is an LPS for the part
  $(q_{i-1},x_{i-1},y_{i-1})\autsteps(q_{i},x_i,y_i)$ such that said run conforms
  to each of these LPSs. By imposing (a)~the linear inequalities of
  \cite{DBLP:journals/jacm/BlondinEFGHLMT21}, which make sure that all counters
  stay non-negative, and (b)~our linear inequalities $\bA\bz < \bb$ and equalities
  $\bC\bx = \bd$, we obtain a new (poynomial-size) system of linear inequalities
  over the exponents in the LPSs.
  
  By \cite{vonzurGathenSieveking1978} systems like these have minimal solutions
  with at most exponential entries, yielding an overall run with at most
  exponential counter values. Binary encoding then means that these solutions
  only take up polynomial space. More specifically, this implies that we can
  simply guess configurations $(q_1,x_1,y_1)$ to $(q_m,x_m,y_m)$ of the
  constrained run in $\PSPACE$, and then check that equalities and inequalities
  of our system hold for them, i.e.\ that they are actual solutions to the
  system. This concludes the description of our decision procedure.
  
  As a final remark, note that \cite{vonzurGathenSieveking1978}
  assumes inequalities of the form $t_1 \leq t_2$ rather than $t_1 < t_2$.
  However, since we seek solutions in $\N^m$, we can simply express $t_1 < t_2$
  as $t_1 + 1 \leq t_2$ to circumvent this issue. 
\end{proof}

\section{Regular Separability vs.\ Intersection}\label{appendix-intersection}
In this section we prove the second part of \cref{Theorem:CounterExamples}.
To this end we present a class of WSTS such that, for their $\omega$-languages, intersection is decidable whereas regular separability is not. 
A $d$-dimensional \emph{reset B\"{u}chi VASS} over alphabet $\Sigma$ is a tuple $\cV=(Q,q_0,T,F)$. 
The only difference to B\"uchi VASS is in the finite set of transitions which, besides adding a vector, may reset a counter,  $T\subseteq
Q\times(\Z^d\cup \{\reset_1,\ldots,\reset_d\})\times\Sigma^*\times Q$. 
The configurations are defined like for B\"uchi VASS, but the transition relation has to be adapted.  
We have $(q,\bmm) \xrightarrow{w} (q',\bmm')$ if there is a transition
$(q,x,w,q')$ such that either
\begin{myenum}
  \myitem $x\in\Z^d$ and $\bmm'=\bmm+x$ or
  \myitem $x=\reset_i$ for some $i\in[1,d]$ and $\bmm'(j)=\bmm(j)$ for $j\in[1,d]\setminus\{i\}$ and $\bmm'(i)=0$.
\end{myenum}
Acceptance is defined as before, and so is the language (of infinite words) $L(\cV)$. 

For general B\"{u}chi reset VASS, emptiness and intersection are undecidable~\cite[Theorem~10]{Mayr2003}. 
We consider a slight restriction of the model that makes the problems decidable.
A B\"{u}chi reset VASS is \emph{weak} if there is no path from a final state to a reset transition. In particular, an
accepting run can only perform finitely many resets. Note that the usual
product construction of $\cV_1$ and $\cV_2$ to yield a B\"{u}chi reset VASS
for $L(\cV_1)\cap L(\cV_2)$ preserves weakness. 
\begin{theorem}
  For weak B\"{u}chi reset VASS, emptiness (hence intersection) is decidable.
\end{theorem}
Here, emptiness can be decided using standard techniques. We order the
configurations $Q\times\N^d$ in the usual way: We have $(q,\bmm)\le (q',\bmm')$
if $q=q'$ and $\bmm\le\bmm'$. First, one observes that for any
B\"{u}chi VASS $\cV$, the set $U(\cV)\subseteq Q\times\N^d$ of all
configurations $(q,\bmm)$ from which an infinite accepting run can start, is
upward closed. Moreover, using a saturation procedure, we can effectively
compute the finitely many minimal elements
$(q_1,\bmm_1),\ldots,(q_\ell,\bmm_\ell)$ of $U(\cV)$.
The details can be found in \cref{compute-min-conf} at the end of this section.
Then, for a weak
B\"{u}chi reset VASS $\cV$, we do the following. We construct the B\"{u}chi
VASS $\cV'$, which is obtained from $\cV$ by deleting all reset transitions.
Now $L(\cV)$ is non-empty if and only if $\cV$, as a reset VASS, can cover any
of the configurations $(q_1,\bmm_1),\ldots,(q_\ell,\bmm_\ell)$ of $U(\cV')$.
Whether the latter is the case can be decided because coverability is decidable
in reset
VASS~\cite{DBLP:conf/icalp/DufourdFS98,DBLP:journals/tcs/FinkelS01,DBLP:journals/iandc/AbdullaCJT00}.

\begin{theorem}\label{regsep-undecidable-wbrvass}
  For weak B. reset VASS over $\Sigma_1$, regular separability from $D_1$ is undecidable.
\end{theorem}
We reduce from the place boundedness problem for reset VASS.
A reset VASS is a B\"{u}chi reset VASS without input
words and without final states. 
For $k\in\N$, we say that a reset VASS $\cV$ is \emph{$k$-place bounded}
if for every reachable configurations $(q,\bmm)$, we have $\bmm(1)\le k$. 
Moreover, we call $\cV$ \emph{place bounded} if $\cV$ is $k$-place bounded for some $k\in\N$.
The \emph{place boundedness problem} then asks whether a given reset VASS is place bounded.
The place boundedness problem (more generally, the boundedness problem) for
reset VASS is known to be undecidablei~\cite[Theorem
8]{DBLP:conf/icalp/DufourdFS98} (for a simpler proof, see \cite[Theorem
18]{Mayr2003}).

\begin{figure}[t]
\begin{center}
\scalebox{0.9}{
\begin{tikzpicture}[initial text={}]
\node[state,initial] (q0) {$q_0$};
\node[state,right=1cm of q0] (q) {$q$};
\node[state,right=1cm of q] (q1) {$q_1$};
\node[state,right=1cm of q1,accepting] (q2) {$q_2$};
\node[state,right=3cm of q2] (q3) {$q_3$};
\path[->] 
(q) edge node[above] {$\varepsilon$} (q1)
(q1) edge [loop above] node {$\bzero|a_1$} (q1)
(q1) edge [loop below] node {$\bzero|\bar{a}_1$} (q1)
(q1) edge node[above] {$\varepsilon$} (q2)
(q2) edge[bend left=15] node[above] {$\bzero|\varepsilon$} (q3)
(q3) edge[bend left=15] node[below] {$\bzero|\bar{a}_1$} (q2)
(q3) edge[loop above] node {$-\be_1 + \be_{d+1}|a_1$} (q3)
(q3) edge[loop below] node {$\be_1 - \be_{d+1}|\bar{a}_1$} (q3)
;
\node (vp) at ($(q0.east)+(0.5cm,0)$) {$\cV_\varepsilon$};
\draw[dashed,rounded corners=6pt,color=black!70,thick] ($(q0.west)+(-0.2,-0.5)$) rectangle ( $(q.east)+(0.2,0.5)$ ); 
\end{tikzpicture}}
\end{center}
\vspace{-0.5cm}
\caption{Weak B\"{u}chi reset VASS $\cV'$ in the proof of \cref{regsep-undecidable-wbrvass}.}\label{figure:weak-buechi-reset-vass}
\end{figure}
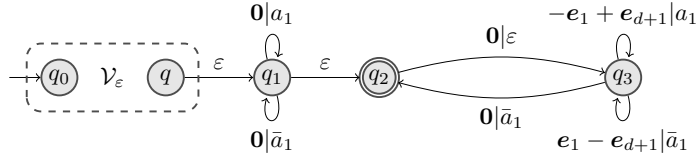
Given a $d$-dim.\ reset VASS $\cV$, we build a $(d+1)$-dim.\ weak B\"{u}chi reset VASS $\cV'$ with
\begin{align}
  L(\cV')=\{ w\in S_{1,k} \mid k\in\N,~\text{$\cV$ can reach some $(q,\mathbf{m})\in Q\times\N^d$ with $\bmm(1)\ge k$}\}.\label{weak-buechi-reset-vass-language}
\end{align}
Before we describe $\cV'$, observe that $\regsep{L(\cV')}{D_1}$ iff
$\cV$ is place bounded. If $\cV$ is $k$-place bounded, then
$L(\cV')\subseteq S_{1,k}$ and thus $\regsep{L(\cV')}{D_1}$. On the
other hand, if $\cV$ is not place bounded, then
$L(\cV')=\bigcup_{k\in\N} S_{1,k}$. As for the B\"{u}chi VASS in
\cref{Figure:examples} (left), one can show $\notregsep{L(\cV')}{D_1}$.

The construction is depicted in \cref{figure:weak-buechi-reset-vass}.  
The dashed box contains
$\cV_\varepsilon$, which is obtained from $\cV$ by changing every
transition $(p,\bu,q)$ into $(p,(\bu,0),\varepsilon,q)$. In the figure, $q$
stands for arbitrary states of $\cV_\varepsilon$, meaning for every
state $q$ in $\cV_\varepsilon$, we have a transition
$(q,\bzero,\varepsilon,q_1)$.  Observe that in the states
$q_1,q_2,q_3$, $\cV'$ behaves exactly like the B\"{u}chi VASS in
\cref{Figure:examples}(left), except that the additional counter
ensures that for each infix the balance on letter $a_i$ is bounded by $k$
from configurations $(q_1, (k, \bu))$.
Thus the accepted language from $(q_1, (k, \bu))$ is exactly $S_{1,k}$.
This shows that $\cV'$ accepts the language
\eqref{weak-buechi-reset-vass-language}.

\begin{lemma}\label{compute-min-conf}
Let $\cV$ be a B\"uchi VASS. We can compute the set $U(\cV)$ of minimal configurations from which there is an infinite accepting run. 
\end{lemma}
\begin{proof}
It is decidable whether a given B\"uchi VASS has an accepting run~\cite{Esparza94,Habermehl97}.
We strengthen this result to checking whether a given B\"uchi VASS $\cV$ has an accepting run starting in a downward-closed set of configurations. 
The downward-closed set is given as a finite union $I$ of ideals, each represented by a generalized configuration $(q, m)\in Q\times (\mathbb{N}\cup\{\omega\})^d$. 
The algorithm is as follows. 
We construct an instrumented B\"uchi VASS $\cV^I$ from $\cV$ and $I$ that starts in a gadget for $I$ from which it moves to $\cV$. 
This gadget selects one of the ideals, say $(q, m)$, and increments each counter $c$ to at most $m(c)$. 
Note that $m(c)$ may be $\omega$, in which case we may put an arbitrary value to this counter.
After this initial phase, $\cV^I$ moves to state $q$ of $\cV$. 
The states in the gadget are not accepting, so $\cV^I$ will eventually move to $\cV$ to obtain an infinite accepting run.  
To be precise, we have in $\cV$ an accepting run from a configuration in $I$ if and only if $\cV^I$ has an accepting run.

With this, we can saturate a set of markings $S$, initially $S=\emptyset$.
We repeatedly ask for an accepting run starting in a downward-closed set of configurations represented by a set of ideals $I$. 
Initially, we just ask for any run, $I=Q \times \{\omega^d\}$.
If such a run does not exist, we return $S$. 
If such a run exists, we can reconstruct a configuration $(q, m)\in I$, $m\in\mathbb{N}^d$,  with which $\cV^I$ moved from the gadget for $I$ to $\cV$. 
This can be done with an enumeration.  
We add $(q, m)$ to $S$ and refine the downward-closed set represented by $I$ by subtracting the upward-closure of the new $S$.   
The subtraction can be computed effectively and yields a new set of ideals with which we repeat the check of an accepting run.  
The process terminates: the set $S$ represents an upward-closed set of configurations, and every infinite sequence of such sets becomes stationary due to the wqo. 
In the moment when the set becomes stationary, we will no longer find an accepting run and return. 
\end{proof}

\label{endofdocument}
\newoutputstream{pagestotal}
\openoutputfile{main.pgt}{pagestotal}
\addtostream{pagestotal}{\getpagerefnumber{endofdocument}}
\closeoutputstream{pagestotal}
\end{document}